\documentclass{article}
\usepackage{authblk}
\usepackage{hyperref}

\usepackage{Definitions}
\usepackage{tcolorbox}
\usepackage{booktabs}

\usepackage{fullpage}
\usepackage[numbers, sort&compress]{natbib}

\usepackage{MnSymbol}

\newtheorem{model}{Model}

\crefname{model}{Model}{Models}
\Crefname{model}{Model}{Models}

\usepackage{graphicx} 

\title{Query-Limited Community Recovery in Stochastic Block Models}
\author[1]{Sabyasachi Basu}
\author[2]{Manuj Mukherjee}
\author[3]{Lutz Oettershagen}
\author[4]{Suhas Thejaswi}

\affil[1]{Microsoft Research, India \protect\\
\texttt{firstname.lastname@microsoft.com} \vspace{0.5em}}

\affil[2]{Indraprastha Institute of Information Technology, India \protect\\
\texttt{firstname@iiitd.ac.in} \vspace{0.5em}}

\affil[3]{University of Liverpool, United Kingdom \protect\\
\texttt{firstname.lastname@liverpool.ac.uk} \vspace{0.5em}}

\affil[4]{Aalto University, Finland\protect\\ 
\texttt{firstname.lastname@aalto.fi}}

\date{}

\begin{document}

\maketitle
\footnotetext{Authors are listed in alphabetical order.}

\begin{abstract}
We study exact community recovery in the two-community stochastic block model on $n$ vertices under limited and noisy access to network data. The learner may query a noisy neighborhood oracle that reveals each true neighbor of a queried vertex independently with fixed probability and never returns non-neighbors, subject to a finite query budget. We consider both oracle-only access and a combined model where the learner also observes a single subsampled copy of the underlying graph. For oracle-only access, balanced uniform querying gives a sharp non-adaptive benchmark: when each vertex is queried the same integer number of times, the observations reduce to an SBM with attenuated edge probabilities and the Abbe-Bandeira-Hall exact-recovery threshold applies. We show that this benchmark is not adaptively optimal: a two-stage adaptive strategy succeeds with $n+o(n)$ queries in a regime where balanced uniform querying requires $m n$ queries for some $m>1$. With an additional subsampled graph, we prove a sublinear-query adaptivity gap: balanced data-independent uniform querying with a sublinear budget does not improve over the subsampled graph alone, whereas adaptive querying can target a small set of uncertain vertices and achieve exact recovery. Thus adaptive data acquisition can strictly improve the information-theoretic limits of exact recovery.
\end{abstract}

\xhdr{Keywords} Stochastic block model, Community detection, Adaptive querying, Query complexity

\section{Introduction}

Recovering latent community structure from network data is a fundamental problem in machine learning, statistics, and network science, with applications ranging from social and communication networks to biological and information systems~\citep{jin2021survey}. The stochastic block model (SBM) is a canonical random-graph model for this task: in its simplest two-community form each vertex (among $n$ vertices) belongs to one of two communities, an edge between two vertices is present with probability $p$ if they are in the same community and with probability $q$ otherwise. This model, introduced by Holland, Laskey, and Leinhardt~\citep{holland1983stochastic}, and widely studied since, has become a standard framework for understanding the statistical and algorithmic limits of community detection. A central task here is \emph{exact recovery}: reconstructing all community labels, up to a global label flip, from the observed graph. In the sparse logarithmic-degree regime, where
$p=\alpha\frac{\log n}{n}$
and
$q=\beta\frac{\log n}{n}$, 
exact recovery exhibits a sharp phase transition governed by the quantity $(\sqrt{\alpha}-\sqrt{\beta})^2$: in particular, recovery is information-theoretically possible precisely above the classical Abbe–Bandeira–Hall threshold~\citep{ABH16}.

In many applications, however, the full graph is not directly observable. Edges may be missing because of subsampling, privacy restrictions, incomplete crawls, or noisy measurements, and additional information may be available only through local probes such as API queries, experiments, or surveys~\citep{handcock2010modeling,larock2020understanding,smith2017network}. Motivated by such settings, prior work has investigated how additional imperfect observations---such as multiple correlated or subsampled graphs---can change the information-theoretic limits of exact recovery~\citep{GRS22,RS21}. In these models, the observation mechanism is exogenous: it is specified ex ante by a fixed probabilistic model, independent of the procedure for exact recovery~\citep{GRS22,RS21,RZ24,abbe2018community}. 
In this paper, we study a complementary source of information. Concretely, we introduce a query-limited model with access to a noisy neighborhood oracle: when queried at a vertex, the oracle reveals each true neighbor independently with probability $w$ and never returns non-neighbors. Subject to a total budget of $k$ queries, the learner must decide which vertices to query and how to allocate its queries. This model is closely related to edge subsampling, since both mechanisms introduce one-sided observation noise. But the key difference is operational: in a subsampled graph, the observation pattern is fixed independently of the learner, whereas in the oracle model the learner chooses where information is acquired, possibly non-uniformly and adaptively.

This leads to a natural question: when can adaptive oracle querying improve exact recovery over uniform querying, and what initial information is required for such an improvement? We investigate this question in two models. In the \emph{oracle-only} model, the learner has no initial graph observation and only has access to a noisy-neighborhood oracle. In the \emph{subsampled-graph-plus-oracle} model, the learner first observes a subsampled graph $G_{\mathrm{sub}}$ and may then use oracle queries to refine the recovery. This two-part formulation separates the case in which adaptivity can only improve constants from the case in which even coarse side information enables a sublinear-query adaptivity gap.

Our work connects exact recovery to adaptive data acquisition for network inference. 
When local measurements are expensive, the learner must decide not only how to infer communities from partial observations, but also how to allocate a limited budget to enable exact recovery.
Such measurements may come from API calls, experiments, surveys, or other procedures that reveal incomplete local neighborhood information. In this case, one can start with uniform querying as a natural baseline, but it need not be the most efficient strategy once a coarse global estimate can identify vertices that remain \emph{difficult} to classify. 
Our results identify regimes in which adaptive querying can selectively target these ``difficult'' vertices and thereby improve the information-theoretic limits of exact recovery. In this sense, our work extends beyond exogenous observation settings studied in the literature on exact-recovery in the SBM with correlated side information, i.e., where the observation mechanism is fixed in advance, by making (adaptive) data acquisition itself part of the inference problem.

\medskip \noindent
\textbf{Our contributions are as follows:}\footnote{Omitted details and full proofs are available in the appendix.}

\smallskip
\xhdr{($i$) A budgeted noisy-neighborhood oracle model}
We introduce a query-limited noisy-neighborhood-oracle model for the logarithmic-degree two-community SBM. The oracle gives local observations with one-sided error, but unlike passive subsampling, the learner controls where those observations are made and may allocate queries adaptively. Thus exact recovery depends not only on how much information is available, but also on how it is acquired.

\xhdr{($ii$) Oracle-only access: a linear-regime adaptivity gap}
We first analyze the oracle-only setting. For balanced uniform querying, where each vertex is queried the same (integer) number of times, the observed graph is again an SBM with attenuated edge probabilities. This gives a sharp threshold for uniform sampling via the Abbe-Bandeira-Hall criterion, with effective edge-retention factor
$
\theta_{\mathrm{unif}}=1-(1-w)^{2k/n}.
$
We then show that uniform querying is not necessarily optimal. For an explicit choice of parameters, a two-stage adaptive strategy first constructs a rough global sketch, identifies a sublinear set of uncertain vertices, and then spends additional queries only on that set. This strategy succeeds with $n+o(n)$ queries, whereas balanced uniform querying requires $x_\star n$ queries for some $x_\star>1$.
We also show that every successful oracle-only strategy requires $\Omega(n)$ queries, so the adaptive construction matches the optimal query order while improving the uniform constant.

\xhdr{($iii$) Subsampled graph plus oracle access:  a sublinear adaptivity gap} We next consider the setting in which the learner is also given a subsampled copy of the graph. In this case, the subsampled graph provides an initial global sketch before any oracle queries are made. We show that, with a sublinear budget $k=o(n)$, uniform oracle querying does not improve the exact-recovery threshold beyond that of the subsampled graph alone. In contrast, a two-stage adaptive strategy can use $G_{\mathrm{sub}}$ to identify a small set of uncertain vertices and spend the sublinear query budget exclusively on them. For suitable parameters, this achieves exact recovery in a regime where uniform querying with the same budget fails, yielding a genuine adaptivity gap.

\begin{table}[t]
\centering
\caption{Summary of passive baselines and our query-limited results. The \emph{Acquisition} column indicates whether the observation pattern is fixed in advance (Exogenous), chosen uniformly independently of the data (Uniform), chosen adaptively from previous observations (Adaptive), or arbitrary (Any). Here
$\Delta^2=(\sqrt{\alpha}-\sqrt{\beta})^2$,
$x$ is the number of queries per vertex, $w$ is the oracle reliability, and
$\theta(x,w)=1-(1-w)^{2x}$
is the induced edge-retention probability.}
\label{tab:main-results}
\resizebox{\textwidth}{!}{%
\begin{tabular}{lllll}
\toprule
\textbf{Setting} & \textbf{Acquisition} & \textbf{\#Queries} &
\textbf{Exact-recovery condition} & \textbf{Ref./Result}\\
\midrule
Standard SBM
& Exogenous
& --
& $\Delta^2>2$
& \cite{ABH16} \\

Correlated graph(s)
& Exogenous
& --
& Modified threshold
& \cite{GRS22,RS21,RZ24} \\

Oracle only
& Uniform
& $xn$
& $\theta(x,w)\,\Delta^2>2$
& Prop.~\ref{prop:oracle-uniform} \\

Oracle only
& Adaptive
& $n+o(n)$
& Below uniform threshold in Ex.~\ref{example1}
& Thm.~\ref{thm:adaptive_repair_screened} (Ex.\ref{example1}) \\

Oracle only
& Any
& $o(n)$
& Impossible
& Thm.~\ref{theorem:sublinear-oracle-impossible} \\

Correlated graph + oracle
& Uniform
& $o(n)$
& Fails in the gap regime
& Thm.~\ref{thm:adaptive-main}\\

Correlated graph + oracle
& Adaptive
& $o(n)$
& Succeeds in the gap regime
& Thm.~\ref{thm:adaptive-main}\\
\bottomrule
\end{tabular}}
\vspace{-2mm}
\end{table}

\smallskip
\Cref{tab:main-results} summarizes the passive baselines and the query-limited
results proved in this paper.

\xhdr{Technical challenges}
The main technical difficulty is that adaptive querying breaks the homogeneous observation structure used in standard exact recovery analyses for the SBM. In the adaptive setting, the probability that an edge is observed may depend on non-uniform, random, and transcript-dependent query counts. Moreover, targeted querying creates a selection effect: a vertex is selected for additional queries precisely because the observations already available around that vertex look ambiguous. Those same observations therefore cannot be reused as unbiased evidence in the final local test.

\emph{Leave-one-out-screening.} 
We overcome these challenges through a leave-one-out screening framework, which plays a fundamentally different role from similar arguments used in prior SBM and random-matrix analyses~\citep{abbe2020entrywise,deng2021strong,gaudio2023community,zhang2024fundamental}. In those works, leave-one-out is primarily a device for controlling dependencies in passive observation models. In our setting, it is the mechanism that makes adaptive data acquisition analyzable. For each vertex $v$, we compute a reference labeling $\hat\sigma^{(-v)}$ using observations that exclude edges incident to $v$, use one comparison set to decide whether $v$ should be selected for refinement, and reserve a disjoint comparison set for the final likelihood test with fresh oracle observations. A key step in our analysis is to prove that the resulting screened set is sublinear, so that only a vanishing fraction of vertices require refinement. This separation decouples adaptive selection from local refinement and makes it possible to concentrate the oracle budget on the sublinear vertex subset responsible for the remaining recovery errors and turns leave-one-out from a technical dependence-control tool into a structural ingredient of the adaptive recovery procedure.

\xhdr{Related Work}\label{sec:related-work} 
SBMs are a central model for studying community recovery, following the formulation of \citet{holland1983stochastic} and the exact-recovery theory developed in subsequent work~\citep{ABH16,abbe2018community,banks2016information,ke2018information,decelle2011asymptotic}. 
The SBM literature spans sharp recovery thresholds, optimal algorithms, and numerous model extensions, including general multi-community limits, spectral and belief-propagation methods, degree-corrected and hierarchical models, attributed graphs, and active querying, e.g., ~\citep{abbe2015recovering,chin2015stochastic,mossel2015consistency,karrer2011stochastic,peixoto2014hierarchical,bothorel2015clustering,gadde2016active}.
We focus here on the lines of work most closely related to partial observation and adaptive data acquisition.
Recovery from incomplete network data has been studied in settings with sampled vertices, missing edges, and non-random missingness~\citep{handcock2010modeling,smith2017network}. These models motivate partial observability, but the sampling mechanism is typically fixed independently of the learner. Our setting differs by making observation itself a decision: the learner allocates a finite budget of noisy neighborhood queries.

The closest adaptive-observation work is \citet{yun2014community}, who studied random and adaptive sampling for community detection and showed that adaptivity can reduce sample complexity in some regimes. Related active-learning approaches for graph clustering and community detection also emphasize query design~\citep{bodo2011active,kushnir2020active,wu2024robust}. These works do not characterize exact recovery in the logarithmic-degree SBM under a one-sided noisy neighborhood oracle.
Our model is also related to correlated SBMs and multiple graph observations, where additional correlated networks can alter exact-recovery thresholds~\citep{RS21,GRS22,RZ24}. Recent computational work further shows that correlated observations introduce new algorithmic phenomena~\citep{chen2026computational,ding2025low,chai2024efficient}. In these settings, however, the observations are fixed in advance; in ours, the learner chooses where to acquire additional local information.
Other SBM extensions, including bipartite, attributed, multiplex, and testing problems, show how richer observation structures affect inference~\citep{yen2020community,nguen2024network}. None address the query-limited noisy-oracle setting studied here. Our contribution is to distinguish regimes where adaptive querying gives no information-theoretic benefit from regimes where coarse side information enables a genuine adaptivity gap.\looseness=-1

\section{Preliminaries} \label{sec:preliminaries}

We consider a stochastic block model (\SBM) on $n$ vertices with two communities. Throughout the paper, we assume that $n$ is even, as is standard practice in the SBM literature (see \eg, \citep{ABH16,mossel2015consistency}). A random graph $G=(V,E)$ with $|V|=n$ is said to be an $\SBM(n,p,q)$, for some $p,q\in[0,1]$, if there exists a partition
$
V=V^+\cup V^-,
$
with
$
|V^+|=|V^-|=n/2,
$
such that, conditional on the partition, edges are generated independently and $(i,j)$ is present with probability $p$ if $i$ and $j$ belong to the same community, and with probability $q$ if $i$ and $j$ belong to different communities. The partition is generated uniformly at random among all balanced bipartitions of $V$.
We denote the ground-truth community labeling by
$\sigmastar\in\{\pm 1\}^n$, where $\sigmastar_i=+1$ if $i\in V^+$ and
$\sigmastar_i=-1$ if $i\in V^-$. Thus
$
\sum_{i\in V}\sigmastar_i=0.
$
Exact recovery is always understood up to a global sign flip.\looseness=-1

Let $A$ denote the adjacency matrix of $G$, where $A_{ij}=1$ ($A_{ij}=0$ resp.) indicate that $(i,j) \in E$ ($(i,j) \notin E$ resp.). Following the notation, we say $G \sim \SBM(n,p,q)$ if, for all $i \neq j$, $\Pr(A_{ij}=1| \sigmastar_i = \sigmastar_j) = p$ and $\Pr(A_{ij}=1 | \sigmastar_i \neq \sigmastar_j) = q$. For a positive integer $k\in\mathbb{N}$, we denote $[k]:=\{1,\dots,k\}$, and for convenience, we consider the vertex set to be $V=[n]$.
For $\alpha,\beta\ge 0$, we use $\Deltasqab{\alpha}{\beta}:=\SBMsigval{\alpha}{\beta}$,  and $\Dplus{\alpha}{\beta}:=\frac{1}{2}\Deltasqab{\alpha}{\beta}$. The ABH exact-recovery threshold for the balanced two-community SBM with edge probabilities $\alpha \frac{\log n}{n}$ and $\beta \frac{\log n}{n}$ is $\Dplus{\alpha}{\beta}>1$, or equivalently $\Deltasqab{\alpha}{\beta}>2$. When the ambient SBM parameters are $\alpha,\beta$, we use  $\Deltasq := \Deltasqab{\alpha}{\beta}$.

A community recovery algorithm takes a graph $G$ as input and outputs an estimate $\sigmatil\in\{\pm 1\}^n$ of the community labels, without the knowledge of the ground truth $\sigmastar$. The success of the algorithm is measured by the \emph{overlap} between $\sigmatil$ and $\sigmastar$. And it is the normalized difference between the number of correctly and incorrectly labeled vertices. Formally, the overlap is expressed as
\begin{equation}\label{eq:overlap}
    \overlap(\sigmastar,\sigmatil):=\frac{1}{n}\left|\sum_{i\in V}\sigmastar_i\,\sigmatil_i \right|.
\end{equation}
Note that, $\overlap(\sigmastar,\sigmatil) \in [0,1]$ and larger values of $\overlap(\sigmastar, \sigmatil)$ indicate better estimate. The \emph{exact recovery} occurs when $\sigmatil=\sigmastar$ which is equivalent to $\overlap(\sigmastar,\sigmatil)=1$.

\noindent
\textbf{Observation models:} We study exact recovery under two related observation models. Both are based on
the same underlying graph
$
G\sim \SBM(n,\alpha \frac{\log n}{n},\beta \frac{\log n}{n}),
$
but they differ in what information is available before the learner starts
querying the oracle.

\begin{model}[Oracle-only access] \label{problem1}
Consider a two-community stochastic block model $\SBM(n,\alpha \frac{\log n}{n},\beta \frac{\log n}{n})$. The learner does not observe the graph directly, but may make at most $k$ queries to a noisy neighborhood oracle. When queried at a vertex, the oracle reveals each true neighbor independently with probability $w$ and never reports non-neighbors. The learner may choose queries non-adaptively or adaptively as a function of previous oracle responses. The goal is exact recovery of the community labels from the oracle transcript.
\end{model}

\begin{model}[Subsampled graph plus oracle access] \label{problem2}
Consider the same underlying SBM, but suppose the learner first observes a subsampled graph $G_{\mathrm{sub}}$, obtained by retaining each edge independently with probability $s\in[0,1]$. After observing $G_{\mathrm{sub}}$, the learner may make at most $k$ queries to the noisy neighborhood oracle that reveals each true neighbor from the underlying graph $G$ with probability $w$. The goal is exact recovery from the combined observation $(G_{\mathrm{sub}},\text{oracle transcript})$.
\end{model}

The main question is how the query allocation affects the exact-recovery threshold. We compare balanced uniform querying, in which the budget is spread evenly over all vertices, with adaptive or non-uniform querying, in which the allocation may depend on the information observed so far.

\section{Community Recovery with Oracle Only Information} \label{sec:oracle}

In this section, we investigate \Cref{problem1}, where the learner has no initial graph observation and must rely entirely on the noisy neighborhood oracle. In \Cref{subsec:oracle-uniform}, we start with balanced uniform querying, which provides the natural (non-adaptive) benchmark and yields a sharp exact-recovery threshold. In \Cref{subsec:oracle-adaptive}, we turn to strategies where the learner may choose to adapt the querying strategy based on the information revealed during the querying process. This comparison isolates the role of adaptivity in the oracle-only model. Our results show that a linear query budget is necessary for exact recovery even under adaptive querying.
At the same time, adaptivity can still be strictly beneficial at the level of constants: for suitable SBM parameters, a two-stage adaptive strategy succeeds with fewer linear-order queries than balanced uniform querying.

\subsection{Uniform Oracle Access} \label{subsec:oracle-uniform}

For a sharp-homogeneous SBM benchmark, we impose a divisibility convention \(k/n \in \mathbb{N}_0\) and each vertex is queried exactly $m := k/n$ times. Under this policy, the revealed graph is again an SBM with attenuated edge probabilities, so the classical Abbe-Bandeira-Hall threshold for exact recovery applies directly. More precisely, an existing edge $\{i,j\}$ is missed by all oracle queries at both endpoints with probability $(1-w)^{2m} = (1-w)^{2k/n}$, and is therefore revealed with probability 
$\theta_{\mathrm{unif}}(k,w) := 1-(1-w)^{2k/n}.$
Consequently, the observed graph is exactly an SBM with effective parameters
$
\tilde \alpha = \alpha \cdot \theta_{\mathrm{unif}}(k,w)\,\text{and}\,
\tilde \beta = \beta \cdot \theta_{\mathrm{unif}}(k,w).
$
Applying the classical exact-recovery threshold gives the following benchmark.
\begin{restatable}[Uniform oracle access benchmark]{proposition}{proporacleuniform} 
\label{prop:oracle-uniform}
Let $G\sim\SBM(n,\alpha\log n/n,\beta\log n/n)$. Suppose $k/n\in\mathbb N_0$
and the learner uses the uniform non-adaptive oracle policy with reliability
$w\in[0,1]$. Then exact recovery from the observed graph is possible w.h.p. if
and only if
\[
\theta_{\rm unif}(k,w)(\sqrt\alpha-\sqrt\beta)^2>2,
\qquad
\theta_{\rm unif}(k,w):=1-(1-w)^{2k/n}.
\]  
\end{restatable}
Thus balanced uniform querying succeeds exactly when the attenuated graph satisfies the threshold 
$\theta_{\mathrm{unif}}(k,w)(\sqrt{\alpha}-\sqrt{\beta})^2 > 2.$
The proposition immediately yields the corresponding requirement for the query budget.
Next, in \Cref{cor:oracle-kmin}, we show that whenever balanced uniform querying succeeds, it requires a \emph{linear number of queries} for exact recovery.

\begin{restatable}{corollary}{cororacleuniformkmin}
\label{cor:oracle-kmin}
Assume $w \in(0,1)$ and $\Deltasq>2$. Under the divisibility convention $k=mn$, exact recovery under uniform oracle access is possible if and only if
$1-(1-w)^{2m}>\frac{2}{\Deltasq}$.
Equivalently, the minimum number of queries per vertex is
\[
m_{\min} := \min\left\{ m\in\mathbb N_0: 1-(1-w)^{2m}>\frac{2}{\Deltasq}\right\},
\]
and the corresponding minimum budget is $k_{\min}=n m_{\min}.$
In the continuous relaxation, the threshold corresponds to
$
m > \frac{\log\!\left(1-\frac{2}{\Deltasq}\right)}{2\log(1-w)}.
$
In particular, \(k_{\min}=\Omega(n)\).
\end{restatable}

\emph{Remark (non-divisible budgets):}
If \(k/n \notin \mathbb{N}_0\), a balanced allocation queries each vertex either \(\lfloor k/n\rfloor\) or \(\lceil k/n\rceil\) times. The resulting observation is not exactly homogeneous, but it is bracketed by the two neighboring homogeneous allocations. Thus \Cref{prop:oracle-uniform} remains the correct benchmark up to a one-query-per-vertex rounding window. See \Cref{app:sec:oracle} for precise arguments.

\subsection{Adaptive Oracle Access} \label{subsec:oracle-adaptive}

We now turn to adaptive querying for exact-recovery in \Cref{problem1}. Here the main idea, 
summarized in Algorithm~\ref{alg:adaptive-screening}, 
is to separate recovery into two stages. In the first stage, the learner uses part of the query budget to construct a coarse global sketch of the communities, and, based on this estimate, produces a preliminary labeling and identifies a small subset of vertices whose labels remain uncertain. In the second stage, the learner spends the remaining queries only on this uncertain set and uses the fresh observations to repair the remaining mistakes. In this section, we formalize when such a two-stage strategy succeeds and show that, for suitable parameters, it can achieve exact recovery with fewer linear-order queries than balanced uniform querying.

\begin{algorithm}[t]
\caption{Two-stage adaptive screening strategy}\small
\label{alg:adaptive-screening}
\DontPrintSemicolon
\KwIn{Oracle access with reliability $w$, first-stage budget $k_0$, refinement parameter $L$}
\KwOut{Community labels $\widehat\sigma$}

\tcp{Stage 1}
Query vertices according to a uniform first-stage allocation using budget $k_0$.\;
From first-stage transcript, compute preliminary labeling $\widehat\sigma^{(0)}$ 
and confidence scores $\forall v\in V$.\;
\tcp{Screening}
Let $A$ be the set of vertices whose first-stage local margin is too small.\;
Keep the labels $\widehat\sigma^{(0)}_v$ for vertices $v\notin A$.\;

\tcp{Stage 2}
\ForEach{$v\in A$}{
    Query $v$ an additional $L$ times.\;
    Compare the fresh neighbors of $v$ against a leave-one-out reference labeling.\;
    Assign $v$ by the resulting local likelihood test.\;
}

\Return the combined labeling: first-stage labels on $A^c$ and refined labels on $A$.\;
\end{algorithm}

\xhdr{A screened-initialization condition}
In this first stage, fix $x_0>0$ and perform $k_0=x_0 n+ o(n)$ uniform non-adaptive oracle queries. The resulting transcript is said to admit a leave-one-out screened initialization if it produces these global and local per-vertex auxiliary objects as follows. Globally, it outputs a preliminary labeling $\widehat \sigma^{(0)}$ and a screened set $\mathcal A \subseteq [n]$ whose labels need to be repaired. Locally, for each vertex $v \in [n]$, it produces a leave-one-out reference labeling $\widehat \sigma^{(-v)}$ together with two disjoint comparison sets $B_v \subseteq [n] \setminus \{v\}$ and $C_v \subseteq [n] \setminus \{B_v \cup \{v\} \}$. 
The labeling $\widehat \sigma^{(-v)}$ is computed without using the first-stage (edge) observations incident to $v$, and it is used only when classifying that same vertex $v$ in the second stage. The screening set $C_v$ is used only to decide whether $v \in \mathcal A$, while the reference set $B_v$ is reserved for the second-stage likelihood test for $v$. This separation between $B_v$ and $C_v$ is imposed to avoid reusing the same incident observations from $v$ both for initial screening and for final refinement: the event $\{v \in \mathcal A\}$ may depend on first-stage observations incident to $v$ only through pairs $\{u,v\}$ with $u \in C_v$ and not through pairs with $u \in B_v$. 
Thus the edges from $v$ to $B_v$ are not conditioned on when deciding whether $v$ is screened and remain as uncontaminated evidence for the refinement step. We further require $B_v$, $C_v$, and the labels $\widehat\sigma^{(-v)}_{B_v\cup C_v}$ (i.e., the collection of leave-one-out labels assigned to vertices in $B_v \cup C_v$) to be functions only of first-stage observations not incident to $v$. We formalize these aspects in the following.

\begin{definition}[Leave-one-out screened initialization]\label{def:screening}
For $\gamma\in(0,1)$ and $\zeta\in(0,1]$, we say that a leave-one-out screened initialization is $\gamma$-screening with reference parameter $\zeta$ if with probability $1-o(1)$, there exists a common global sign $s\in\{\pm1\}$ such that
\begin{equation}
\label{eq:screening_condition}
|\mathcal A|\le n^\gamma,
\qquad
\widehat\sigma^{(0)}_u=s\sigma^*_u\quad\text{for all }u\notin \mathcal A,
\end{equation}
and, for every $v\in \mathcal A$,
\begin{align}
\bigl|\{u\in B_v:\widehat\sigma^{(-v)}_u\ne s\sigma^*_u\}\bigr|&=o(n/\log n),\notag \\
|B_v\cap\{u:\sigma^*_u=+1\}|&=\frac{\zeta n}{2}+o(n),\notag \\
|B_v\cap\{u:\sigma^*_u=-1\}|&=\frac{\zeta n}{2}+o(n). \label{eq:reference_condition}
\end{align}
\end{definition}
In other words, \Cref{def:screening} conveys that after the first stage, vertices that may still be mislabeled are confined to a sublinear set $\mathcal A$, while every screened vertex $v \in \mathcal A$ comes with a large reference set $B_v$ whose labels are almost correct and can therefore be used to classify $v$ in the second stage. The parameter $\zeta$ measures how large this reference set is. The screening set $C_v$ is kept separate from $B_v$ so that the information used to decide whether $v$ needs further refinement is different from the information used to determine its final label.

The next result, in \Cref{thm:adaptive_repair_screened}, shows that any first-stage transcript satisfying \Cref{def:screening} can be extended to achieve exact-recovery by querying only the screened vertices in $\mathcal A$.

\begin{restatable}[Adaptive labeling from a screened initialization]{theorem}{adaptiverepairscreened} 
\label{thm:adaptive_repair_screened}
Assume $\Deltasq>2$ and $w\in(0,1]$.  Suppose the first-stage oracle transcript admits a
$\gamma$-screening initialization with reference parameter $\zeta$ for some $\gamma\in(0,1)$
and $\zeta\in(0,1]$.  Choose a fixed integer $L\ge1$ such
that
\begin{equation}
\label{eq:L_condition_oracle_only}
\frac{\zeta}{2}\bigl(1-(1-w)^L\bigr)\Deltasq>\gamma.
\end{equation}
Then there is an adaptive second-stage strategy using at most $L\,|\mathcal A|$ additional oracle
queries that recovers all labels in $\mathcal A$ correctly with probability $1-o(1)$.  Consequently,
the total query budget is
\[
k= k_0 + L\,|\mathcal A|=x_0\,n+o(n)
\]
on the event~\eqref{eq:screening_condition}--\eqref{eq:reference_condition}, and exact
recovery succeeds w.h.p.
\end{restatable}

\begin{proof}[Proof sketch]
On the high-probability screening event, all vertices outside $\mathcal A$ are already correctly labeled, so it remains only to repair the vertices in $\mathcal A$. In the second stage, fix $v \in \mathcal A$, query $v$ an additional $L$ times and use only these fresh oracle observations from $v$ to the reserved reference set $B_v$. For each $u \in B_v$, let $Y_u = 1$ if $u$ is revealed at least once among these $L$ fresh queries to $v$, and $Y_u = 0$ otherwise. We then classify $v$ by a likelihood-ratio test based on the variables $(Y_u)_{u\in B_v}$, using the leave-one-out reference labels $\widehat \sigma^{(-v)}_{B_v}$. Intuitively, the test compares whether the fresh observations from $v$ are more consistent with $v$ belonging to the $+1$ side of the reference population or the $-1$ side. 
Since $L$ fresh queries reveal an existing edge with probability $a_L:=1-(1-w)^L$, this local experiment is a one-vertex SBM test against an almost-correct reference population of size $|B_v|=\zeta n+o(n)$, with approximately $\zeta n/2$ vertices from each community. Therefore, the Bhattacharyya error exponent for one screened vertex is
$
\frac{\zeta}{2}\,a_L\,\Delta^2+o(1) = \frac{\zeta}{2}\bigl(1-(1-w)^L\bigr)\Delta^2+o(1).
$
Hence the error probability for a single screened vertex is at most
$
n^{-\frac{\zeta}{2}(1-(1-w)^L)\Delta^2+o(1)}.
$
Since $|\mathcal A|\le n^\gamma$ and the exponent is strictly larger than $\gamma$, a union bound implies that all vertices in $\mathcal A$ are recovered correctly w.h.p. The additional second-stage budget is $L\,|\mathcal A|=O(n^\gamma)=o(n)$, because $L$ is fixed and $\gamma<1$. 
\end{proof}
A remaining question is whether the screened initialization can be
obtained using fewer queries than balanced uniform exact-recovery requires. The
answer is yes whenever the first-stage budget is below the uniform threshold but
still large enough to produce a screened initialization. Indeed, let $x_\star$ be
the uniform exact-recovery threshold, i.e.,
$\theta(x_\star,w)\Deltasq=2$. If, for some $x_0<x_\star$, the first-stage
uniform oracle graph generated using $x_0n+o(n)$ queries admits a
$\gamma$-screening initialization with reference parameter $\zeta$ for some
$\gamma<1$, and if $L$ satisfies~\eqref{eq:L_condition_oracle_only}, then
\Cref{thm:adaptive_repair_screened} gives exact recovery with
$
    x_0n+L|A| = x_0n+o(n)
$
queries. Since $x_0<x_\star$, this is below the balanced-uniform exact-recovery
threshold by a linear margin $(x_\star-x_0)n-o(n)$.

\xhdr{Achieving screening below the uniform threshold}
Next, we show that the screened-initialization condition can actually be achieved and it is not merely a formal abstraction. 
For an explicit parameter choice, a first-stage uniform query budget that is still below the balanced-uniform exact-recovery threshold already provides two things: a leave-one-out labeling that is accurate up to only $n^{1/4+o(1)}$ errors, and a small screened set containing the remaining vertices with uncertain labels. An important aspect is that, although this first-stage budget is not large enough for exact recovery by itself, it is still large enough for a standard maximum-likelihood argument to produce a sufficiently accurate coarse estimate, and that weaker guarantee is enough for the second-stage repair step. 

Recall that for parameters $a,b>0$, we write
\[
\Dplus{\alpha}{\beta}:=\frac12\Deltasqab{\alpha}{\beta} 
=\frac12(\sqrt \alpha-\sqrt \beta)^2
=\frac{\alpha+\beta}{2}-\sqrt{\alpha\beta}.
\]
For each $v\in[n]$, let $G^{(-v)}$ be the graph
obtained by deleting all edges incident to $v$, and let
\[
\mathcal B_{-v}:=\{\tau\in\{\pm1\}^{[n]\setminus\{v\}}:
|\sum_{u\ne v}\tau_u|=1\}.
\]
Here, $\mathcal{B}_{-v}$ is the set of balanced labelings for the leave-one-out graph.
  The balanced leave-one-out MLE $\widehat\sigma^{(-v)}$ is any maximum-likelihood estimator
for $G^{(-v)}$ over $\mathcal B_{-v}$.  Its two global signs are fixed by the following
deterministic anchor convention.  Let $D_n=\{1,\ldots,\lceil n^{3/4}\rceil\}$.  For each
$v$, orient $\widehat\sigma^{(-v)}$ so that
$
\sum_{u\in D_n\setminus\{v\}}\widehat\sigma^{(-v)}_u\ge0,
$
breaking ties arbitrarily.

\begin{restatable}[Screening estimate for the concrete first stage]{lemma}{concretefirststagescreening}
\label{lem:concrete_first_stage_screening}
Consider the oracle-only model with $\alpha=16$ and $\beta=1$. Let
$\theta_0 := 1-2^{-1/3}$, and $w_0 := 1-\sqrt{1-\theta_0}$.
If the first stage queries every vertex exactly once with oracle reliability
$w_0$, then the first-stage transcript admits a $0.88$-screening initialization
with reference parameter $1/2$, in the sense of
\eqref{eq:screening_condition}--\eqref{eq:reference_condition}.\looseness=-1
\end{restatable}

Combining \Cref{lem:concrete_first_stage_screening} with \Cref{thm:adaptive_repair_screened} gives the following concrete parameter choice, where the first-stage budget is below the balanced-uniform threshold and the adaptive two-stage repair still succeeds. 
\Cref{example1} provides a concrete parameter regime to demonstrate that  adaptivity can improve the constant in the linear query budget.

\begin{example}\label{example1}
Let $\alpha=16$, $\beta=1$ and set the oracle reliability parameter to $w=1-2^{-1/6}$. In the first stage, we query every vertex exactly once, so we use $k_0=n$ queries. The edge-retention factor of the corresponding uniform subsampled graph is $\theta_0=1-(1-w)^2=1-2^{-1/3}$ and the graph is a homogeneous SBM with effective parameters $a= 16\, (1-2^{-1/3})$ and $b=1-2^{-1/3}$. Applying \Cref{lem:concrete_first_stage_screening} on this graph, we obtain a $0.88$-screening initialization with reference parameter $\zeta=1/2$.

The relaxed balanced-uniform exact-recovery threshold is determined by $\theta(x,w)\Delta^2 = 2$, where $\theta(x,w)=1-(1-w)^{2x}$ and $\Delta^2=(\sqrt{\alpha}-\sqrt{\beta})^2=9$. Thus $x_\star$ is the solution of $9\,(1-(1-w)^{2\,x_\star})=2$, precisely $x_\star=3\,\log_2(9/7) > 1$. Since the first stage uses $x_0=1$ query per vertex, it operates strictly below the balanced-uniform exact-recovery threshold.

Now take $L=5$. Using $\zeta=1/2$ and $\Delta^2 = 9$, we obtain 
$\frac{\zeta}{2}(1-(1-w)^L)\Deltasq =\frac14(1-2^{-5/6})\cdot 9 >0.88$.
The adaptive strategy succeeds with the query budget $ k=n+5n^{0.88+o(1)}=n+o(n)$. In contrast, balanced uniform querying requires
$x_\star n = 3\log_2(9/7)n \approx 1.087n$ 
queries. Thus, for this parameter choice and for sufficiently large $n$, %
adaptive querying achieves exact recovery with strictly fewer linear-order queries than uniform querying. $\hfill \Diamond$
\end{example}

Next, we establish a stronger result that this is the strongest possible type of separation in the oracle-only model, i.e., any adaptive strategy still requires a linear number of queries for exact recovery.\looseness=-1

\begin{restatable}{theorem}{sublinearoracleimpossible}
\label{theorem:sublinear-oracle-impossible}
In \Cref{problem1} every oracle-only exact-recovery strategy requires $k=\Omega(n)$ queries.
\end{restatable}
\begin{proof}[Proof sketch.]
We prove the contrapositive. If an oracle-only strategy uses \(k=o(n)\) queries, then only \(o(n)\) vertices are ever queried. Since the underlying SBM has maximum degree \(O(\log n)\) with high probability, the total number of vertex appearances in all oracle responses is only \(o(n\log n)\). Hence there are \(n-o(n)\) vertices that are never queried and appear only \(o(\log n)\) times in the transcript; by exact balance, this set contains linearly many vertices from each community. For any such \(+\)-vertex \(i\) and \(-\)-vertex \(j\), compare the truth with the balanced labeling obtained by swapping \(i\) and \(j\). Since neither vertex is queried and both appear rarely, this swap changes the transcript likelihood by only a small polynomial factor, while preserving the balance constraint. Thus there are linearly many distinct balanced swapped labelings with likelihood comparable to that of the truth, so the posterior mass of the true labeling is \(o(1)\). Consequently even the MAP estimator fails with probability tending to one, and exact recovery is impossible when \(k=o(n)\). Therefore every successful oracle-only strategy requires \(k=\Omega(n)\).
\end{proof}

\section{Community Recovery with Additional Correlated Information} \label{sec:recovery}

We now turn to \Cref{problem2}, where the learner first observes a subsampled graph $G_{\mathrm{sub}}$ and then uses the noisy neighborhood oracle for refinement. The main question is whether this additional correlated side information $G_{\mathrm{sub}}$ makes adaptive querying more effective than uniform querying, and how this contrasts with the oracle-only setting. We show that, under a sublinear budget, uniform querying cannot improve the exact-recovery threshold beyond that of $G_{\mathrm{sub}}$ alone. In contrast, under the same sublinear budget, an adaptive strategy can use $G_{\mathrm{sub}}$ to identify uncertain vertices and concentrate the oracle budget on them, leading to exact recovery. Thus, unlike in the oracle-only model where the adaptivity gap is in linear order, the presence of correlated side information results in a genuine sublinear adaptivity gap.

\subsection{Why does Uniform Querying Under a Sublinear Budget Fail?}
The non-adaptive baseline spreads the oracle budget as uniformly as possible, independently of the observed graph. In the sublinear regime this should be interpreted in the integer balanced sense: only $k=o(n)$ vertices receive an oracle query, while $n-o(n)$ vertices receive no oracle information at all. Thus the combined observation is not exactly a homogeneous SBM with a new edge-retention parameter.
Nevertheless, the conclusion is the same as the homogeneous heuristic would suggest. Since the queried set is data-independent and contains only $o(n)$ vertices, uniform querying cannot target the vertices that are ambiguous under $G_{\rm sub}$. Hence it cannot change the exact-recovery threshold determined by the subsampled graph alone. The following lemma formalizes this statement.

\begin{restatable}[Uniform querying fails]{lemma} {uniformsublinearfails}\label{lemma:uniform_sublinear_fails}
	Consider \Cref{problem2} with subsampling rate $s\in(0,1)$, oracle reliability $w\in(0,1]$, and query budget $k=o(n)$.  Suppose the learner uses a data-independent balanced uniform querying strategy: the queried set is fixed in advance, or chosen independently of $G_{\rm sub}$ and the oracle responses. Then, in the sublinear-budget regime, uniform querying does not improve the exact-recovery threshold beyond that of $G_{\rm sub}$ alone. In particular, exact recovery under uniform querying is impossible w.h.p. whenever $s\,\Deltasq < 2$.
\end{restatable}
\begin{proof}[Proof sketch]
With $k=o(n)$, a balanced uniform strategy queries only $o(n)$ vertices, independently of $G_{\rm sub}$ and of the labels. Hence $n-o(n)$ vertices receive no oracle information at all. For any two unqueried vertices of opposite communities, the oracle can affect their comparison to only the $o(n)$ queried vertices; their evidence from the remaining $n-o(n)$ vertices is exactly the evidence available in $G_{\rm sub}$. Thus the local exact-recovery exponent for a balanced swap of two unqueried opposite-community vertices is the same as for the subsampled graph alone, namely $s\Delta^2$. If $s\Delta^2<2$, the usual ABH second-moment converse gives, with high probability, at least one likelihood-competitive balanced swap among unqueried vertices. Therefore the posterior cannot concentrate on the true balanced labeling, and uniform sublinear querying cannot achieve exact recovery.
\end{proof}

\subsection{Adaptive Targeting Succeeds Under a Sublinear Budget}

The failure in \Cref{lemma:uniform_sublinear_fails} is not necessarily due to the oracle being too weak; rather, it may stem from the fact that uniform querying ignores where the subsampled graph remains (un)certain. Once the querying strategy can exploit the correlated information from $G_{\rm sub}$, this limitation disappears.

\begin{restatable}{lemma}{adaptivesublinearsucceed}\label{lemma:adaptive_sublinear_succeed}
	Consider  \Cref{problem2}, a subsampling rate $s\in(0,1)$, and oracle reliability $w\in(0,1]$. There exist constants $(\alpha,\beta,s,w)$ and a budget $k=o(n)$ such that a two-stage strategy that uses $G_{\rm sub}$ to \emph{target} queries succeeds at exact recovery w.h.p.
\end{restatable}

\begin{proof}[Proof sketch]
Take, for example, $\alpha=16$, $\beta=1$, $w=1/2$, and choose $s=2(1-\epsilon)/9$ for a sufficiently small constant $\epsilon>0$. Then $G_{\rm sub}$ is below the exact-recovery threshold, but it still gives a leave-one-out preliminary labeling with only $n^{\epsilon+o(1)}$ errors. Using this preliminary labeling, compute for each vertex a local likelihood score from a reserved part of its subsampled incident edges, and screen only vertices whose score has small margin. Large-deviation bounds show that the screened set has size at most $n^{0.77}$ w.h.p., while every unscreened vertex is already correctly labeled. Query each screened vertex once and classify it using fresh oracle observations against a disjoint reference set. The oracle likelihood test has error probability $n^{-0.7875+o(1)}$ per screened vertex, so a union bound over $n^{0.77}$ screened vertices succeeds. Thus $k=\lceil n^{0.78}\rceil=o(n)$ queries suffice for exact recovery w.h.p.
\end{proof}

Combining the uniform-querying lower bound with the adaptive construction yields the main separation of this section: with the same sublinear oracle budget, data-independent querying cannot improve on the subsampled graph, whereas adaptive querying achieves exact recovery.

\begin{restatable}{theorem}{adaptivemain}\label{thm:adaptive-main}
In \Cref{problem2}, there exist SBM parameters $\alpha,\beta>0$, a subsampling rate $s\in(0,1)$, an oracle reliability $w\in(0,1]$ and a query budget $k=o(n)$ such that simultaneously: 

\squishlist
    \item[1.] Every data-independent uniform querying strategy fails to achieve exact recovery w.h.p. 
    \item[2.] There exists a two-stage adaptive querying strategy that uses $G_{\rm sub}$ to target queries and achieves exact recovery with high probability.
\squishend
\end{restatable}

\section{Conclusion} \label{sec:conclusion}
We studied exact community recovery in the two-community SBM under budgeted noisy-neighborhood-oracle access. Our models make data acquisition part of the inference problem: the learner must decide not only how to recover communities from partial observations, but also where to acquire those observations.
In \Cref{problem1}, balanced uniform querying gives a sharp non-adaptive benchmark, but a two-stage adaptive strategy can succeed with $n+o(n)$ queries in regimes where uniform querying requires $mn$ queries for some $m>1$. 
In \Cref{problem2}, a subsampled graph can identify locally ambiguous vertices, allowing adaptive querying to focus a sublinear budget and yielding a genuine adaptivity gap. 
Thus, under the same observation budget, different acquisition rules can lead to fundamentally different exact-recovery behavior.

\xhdr{Scope and Limitations}
Our results are asymptotic and apply to the canonical balanced two-community logarithmic-degree SBM with independent one-sided oracle noise. In this context, we isolate the role of adaptive data acquisition and prove concrete information-theoretic gains over uniform querying. However, characterizing these adaptive separations to sharp optimal constants and extending it to a full adaptive phase diagram across all parameter regimes is left open.

\xhdr{Future Work}
By treating data acquisition as part of the inference problem, we open several directions for future work, including design of general adaptive querying frameworks that match or outperform uniform querying across broader SBM regimes, and closing the remaining exact-threshold gaps.\looseness=-1

\xhdr{Societal Impact}
This work is theoretical and does not involve real data, human subjects, or deployed systems. Its potential positive impact is to clarify when limited local measurements can reduce the cost of network recovery. At the same time, improved community recovery from partial observations could be misused in privacy-sensitive network settings. Any practical deployment should therefore account for consent, privacy constraints, and the sensitivity of inferred communities.

\section*{Funding and Conflict of Interest Disclosure}
Suhas Thejaswi acknowledges support from the Technology Industries of Finland Centennial Foundation through a grant awarded to Aalto University. The authors declare that the funding source does not create any conflict of interest in relation to this work.

\bibliographystyle{plainnat}
\bibliography{refs}

\newpage
\appendix

\section{Omitted proofs from \Cref{sec:oracle}} \label{app:sec:oracle}

For the sharp homogeneous-SBM benchmark, we impose the divisibility convention $k/n\in\mathbb N_0$. Writing $m:=k/n$, the uniform policy queries every vertex exactly $m$ times. Hence an existing edge $\{i,j\}$ is missed by all oracle queries at both endpoints with probability $(1-w)^{2m}$, and is therefore revealed with probability:
$
\theta_{\rm unif}(k,w):=1-(1-w)^{2m}
=1-(1-w)^{2k/n}.
$
Consequently, the observed graph is exactly an SBM with effective parameters
\[
\widetilde \alpha=\alpha\theta_{\rm unif}(k,w),
\qquad
\widetilde \beta=\beta\theta_{\rm unif}(k,w).
\]
Applying the Abbe--Bandeira--Hall threshold gives the following benchmark.

\proporacleuniform*

\begin{proof}
Under the uniform policy, each endpoint of an existing edge is queried $m=k/n$
times. Thus the edge is revealed with probability
$\theta_{\rm unif}(k,w)=1-(1-w)^{2m}$, independently across edges conditional on
the labels. Hence the observed graph is an SBM with parameters
\[
\widetilde p=\alpha\theta_{\rm unif}(k,w)\frac{\log n}{n},
\qquad
\widetilde q=\beta\theta_{\rm unif}(k,w)\frac{\log n}{n}.
\]
The Abbe--Bandeira--Hall threshold applied to this attenuated SBM gives
\[
(\sqrt{\widetilde\alpha}-\sqrt{\widetilde\beta})^2>2,
\]
which is equivalent to
\[
\theta_{\rm unif}(k,w)(\sqrt\alpha-\sqrt\beta)^2>2.
\]
\end{proof}

\cororacleuniformkmin*
\begin{proof}
Under the uniform policy, each endpoint of an existing edge is queried $m=k/n$ times. Thus the edge is revealed with probability
$\theta_{\rm unif}(k,w)=1-(1-w)^{2m}$, 
independently across edges conditional on the labels. Hence the observed graph is an SBM with parameters
\[
\widetilde p=\alpha\theta_{\rm unif}(k,w)\frac{\log n}{n},
\qquad
\widetilde q=\beta\theta_{\rm unif}(k,w)\frac{\log n}{n}.
\]
The Abbe--Bandeira--Hall threshold applied to this attenuated SBM gives
\[
(\sqrt{\widetilde\alpha}-\sqrt{\widetilde\beta})^2>2,
\]
which is equivalent to
\[
\theta_{\rm unif}(k,w)(\sqrt\alpha-\sqrt\beta)^2>2.
\]
\end{proof}

If \(w=1\), then \(\theta_{\mathrm{unif}}(k,1)=1\) for every \(m=k/n\ge 1\). Thus, under the divisibility convention, one query per vertex suffices exactly when
\[
(\sqrt{\alpha}-\sqrt{\beta})^2>2,
\]
and in this case \(k_{\min}=n\). If \(w=0\), then \(\theta_{\mathrm{unif}}(k,0)=0\), so no edge is ever revealed and exact recovery is impossible.

\begin{remark}[Non-divisible budgets]
If $k/n\notin\mathbb N$, let $r=\lfloor k/n\rfloor$. A balanced allocation queries each vertex either $r$ or $r+1$ times. Conditional on this allocation, the observed graph is generally not a homogeneous SBM, because $\rho_{ij}=1-(1-w)^{R_i+R_j}$ depends on the endpoint query counts.

The homogeneous theorem above nevertheless gives a useful bracket. The balanced allocation is at least as informative as the homogeneous allocation with $r$ queries per vertex, since the learner may ignore any extra queries. Conversely, it is no more informative than the homogeneous allocation with $r+1$ queries per vertex. Therefore, with
$\theta_\ell:=1-(1-w)^{2\ell}$,
we have
\[
\theta_r(\sqrt{\alpha}-\sqrt{\beta})^2>2
\quad\Longrightarrow\quad
\text{balanced uniform querying succeeds},
\]
whereas
\[
\theta_{r+1}(\sqrt{\alpha}-\sqrt{\beta})^2<2
\quad\Longrightarrow\quad
\text{balanced uniform querying fails}.
\]
Thus the non-divisible case is bracketed by the two neighboring homogeneous
allocations. We discuss the genuinely sublinear-budget regime separately in
\Cref{lemma:uniform_sublinear_fails}.
\end{remark}

\adaptiverepairscreened*

\begin{proof}
On the event~\eqref{eq:screening_condition}--\eqref{eq:reference_condition}, fix the common global sign $s$.  For each $v\in \mathcal A$, the adaptive stage queries $v$ exactly $L$ additional times.  It ignores all responses from vertices outside the reference set $B_v$ and classifies $v$ using only the responses from $v$ to $B_v$ and the labels $\widehat\sigma^{(-v)}_{B_v}$.  By the split measurability condition, the first-stage data that determine $\{v\in \mathcal A\}$ do not use observations on pairs $\{u,v\}$ with $u\in B_v$.  Hence, after conditioning on the first-stage sigma-field that defines $\{v\in \mathcal A\}$, $B_v$, $C_v$, and $\widehat\sigma^{(-v)}_{B_v\cup C_v}$, the edge variables $\{G^*_{uv}:u\in B_v\}$ are still independent conditional on the labels, and the new oracle coins are fresh.

Let
\[
a_L:=1-(1-w)^L.
\]
For $u\in B_v$, define $Y_u=1$ if at least one of the $L$ new queries to $v$ returns $u$, and $Y_u=0$ otherwise.  Conditional on the labels, the variables $\{Y_u:u\in B_v\}$ are independent Bernoulli random variables with parameter $a_Lp$ when $\sigma^*_u=\sigma^*_v$ and $a_Lq$ when $\sigma^*_u\ne\sigma^*_v$.  By \eqref{eq:reference_condition}, $B_v$ contains $\zeta n/2+o(n)$ vertices of each sign and therefore $|B_v|=\zeta n+o(n)$.

We now spell out how the second stage assigns a label to a fixed vertex $v\in\mathcal A$.  The algorithm treats the labels in $B_v$ as a reference set and uses only the fresh observations $Y_u$ from the additional queries to $v$.  The only quantity being decided in this calculation is whether $v$ has the same label as the $+1$ side of the reference set or the $-1$ side.  First suppose that the reference labels are exact, after applying the common global sign.  For a candidate value $\tau\in\{\pm1\}$ of $s\sigma^*_v$, the conditional log-likelihood of the observations $Y=(Y_u)_{u\in B_v}$ is, up to constants independent of $\tau$,
\[
\begin{aligned}
\ell_v(\tau;Y)
=\sum_{u\in B_v}
&\mathbf 1\{\widehat\sigma^{(-v)}_u=\tau\}
\left[Y_u\log(a_Lp)+(1-Y_u)\log(1-a_Lp)\right]\\
&+\mathbf 1\{\widehat\sigma^{(-v)}_u\ne\tau\}
\left[Y_u\log(a_Lq)+(1-Y_u)\log(1-a_Lq)\right].
\end{aligned}
\]
Thus the log-likelihood ratio in favor of $s\sigma^*_v=+1$ over $s\sigma^*_v=-1$ is
\begin{equation}
\label{eq:second_stage_llr_oracle_only}
\Lambda_v^{(2)}
=
\sum_{u\in B_v}\widehat\sigma^{(-v)}_u
\left[
Y_u\log\frac{p}{q}
+(1-Y_u)\log\frac{1-a_Lp}{1-a_Lq}
\right].
\end{equation}
The second-stage decision rule is: set $\widehat\sigma^{(1)}_v=+1$ if $\Lambda_v^{(2)}\ge0$ and $\widehat\sigma^{(1)}_v=-1$ otherwise, with the final labels interpreted up to the same common global sign.  In words, the algorithm compares whether the fresh neighbors revealed by querying $v$ look more like connections to vertices labeled $+1$ or to vertices labeled $-1$ in the reference set.

The corresponding Bhattacharyya coefficient is also explicit.  For one reference vertex, the two hypotheses swap the Bernoulli parameters $a_Lp$ and $a_Lq$, so the one-coordinate coefficient is
\[
b_L
:=
\sqrt{a_Lp\,a_Lq}+\sqrt{(1-a_Lp)(1-a_Lq)}.
\]
Since the coordinates are conditionally independent and $|B_v|=\zeta n+o(n)$, the product Bhattacharyya coefficient is
\[
\operatorname{BC}_v=b_L^{|B_v|}=b_L^{\zeta n+o(n)}.
\]
The likelihood-ratio test has error probability at most this coefficient.

We next return to the actual reference labels. Let $E_v=\{u\in B_v:\widehat\sigma^{(-v)}_u\ne s\sigma^*_u\}$.  By \eqref{eq:reference_condition}, $|E_v|=o(n/\log n)$.  Relative to the exact-reference likelihood ratio, only the terms indexed by $E_v$ are changed.  The non-edge contribution of these terms is at most
\[
|E_v|\left|\log\frac{1-a_Lp}{1-a_Lq}\right|=o(1),
\]
while the number of newly revealed edges from $v$ to $E_v$ has mean $O(|E_v|\log n/n)=o(1)$.  Hence the mislabeled reference vertices perturb \eqref{eq:second_stage_llr_oracle_only} by $o(\log n)$ with probability $1-o(1)$, uniformly over the screened vertices.  Therefore the same Bhattacharyya exponent remains valid up to an $o(\log n)$ error. Using $p=\alpha\log n/n$ and $q=\beta\log n/n$,
\[
\begin{aligned}
b_L
&=\sqrt{a_Lp\,a_Lq}+\sqrt{(1-a_Lp)(1-a_Lq)}\\
&=a_L\sqrt{pq}+1-\frac{a_L(p+q)}{2}+O(p^2+q^2)\\
&=1-\frac{a_L}{2}(\sqrt p-\sqrt q)^2+O\left(\frac{\log^2 n}{n^2}\right)\\
&=1-\frac{a_L\Deltasq}{2}\frac{\log n}{n}
+O\left(\frac{\log^2 n}{n^2}\right).
\end{aligned}
\]
Consequently,
\[
\operatorname{BC}_v
=\left(1-\frac{a_L\Deltasq}{2}\frac{\log n}{n}
+O\left(\frac{\log^2 n}{n^2}\right)\right)^{\zeta n+o(n)}
=n^{-\zeta a_L\Deltasq/2+o(1)}.
\]
Therefore the conditional probability that the likelihood-ratio test misclassifies $v$ is at
most
\[
n^{-\zeta a_L\Deltasq/2+o(1)}.
\]
By~\eqref{eq:L_condition_oracle_only}, choose $\eta>0$ such that
$\zeta a_L\Deltasq/2\ge \gamma+2\eta$.  A union bound over the at most $n^\gamma$ vertices in $\mathcal A$
gives
\[
\Pr(\text{some }v\in \mathcal A\text{ is misclassified and the screening event holds})
\le n^\gamma n^{-\zeta a_L\Deltasq/2+o(1)}
\le n^{-\eta+o(1)}=o(1).
\]
All vertices outside $\mathcal A$ are correct on~\eqref{eq:screening_condition} by assumption, so the combined labeling is correct up to the common global sign $s$ w.h.p.  The query count is
$k_0+L|\mathcal A|=x_0n+O(n^\gamma)=x_0n+o(n)$ on the same event.
\end{proof}

Before we can prove \Cref{lem:concrete_first_stage_screening}, we need the following result.
\begin{restatable}[Uniform leave-one-out rough initialization]{lemma}{roughinitialization}
\label{lem:rough_initialization}
Consider the exactly balanced two-block SBM on an even number $n$ of vertices, with edge
probabilities $a\log n/n$ and $b\log n/n$.  If $\Dplus{a}{b}>3/4$, then the balanced
leave-one-out MLEs satisfy the following with probability $1-o(1)$: there is a common global
sign $s\in\{\pm1\}$ such that, simultaneously for every $v\in[n]$,
\[
\bigl|\{u\ne v:\widehat\sigma^{(-v)}_u\ne s\sigma^*_u\}\bigr|
\le n^{1/4+o(1)}.
\]
In particular, the number of leave-one-out reference-label errors is $o(n/\log n)$ uniformly
over $v$.
\end{restatable}

\begin{proof}
We use a union-bound proof following in the style of ABH, keeping track of the balanced constraint. Fix $v$ and condition on the true labels.  Write $N=n-1$.  The restriction of the truth to $[n]\setminus\{v\}$ belongs to $\mathcal B_{-v}$, since the full truth is exactly balanced. For $\tau\in\mathcal B_{-v}$, let $L_{-v}(\tau)$ denote the SBM likelihood of the graph on $[n]\setminus\{v\}$ under the labeling $\tau$; that is, the product likelihood over pairs $\{i,j\}\subseteq[n]\setminus\{v\}$ with edge probability $a\log n/n$ when $\tau_i=\tau_j$ and $b\log n/n$ otherwise.

For $\tau\in\mathcal B_{-v}$, choose its global sign so that
\[
m=m(\tau):=\bigl|\{u\ne v:\tau_u\ne\sigma^*_u\}\bigr|\le N/2.
\]
Let $S$ be this disagreement set.  A pair $\{i,j\}\subseteq[n]\setminus\{v\}$ has different within/between-community status under $\tau$ and under $\sigma^*$ exactly when one endpoint lies in $S$ and the other lies outside $S$.  Thus the likelihood ratio between $\tau$ and the truth only involves
\[
M(S)=m(N-m)
\]
pairs; all other pair contributions cancel.  For each such pair, the two hypotheses swap the Bernoulli parameters $a\log n/n$ and $b\log n/n$.  Hence, by the Bhattacharyya bound applied to the likelihood-ratio test between the truth and this fixed alternative,
\[
\Pr\{L_{-v}(\tau)\ge L_{-v}(\sigma^*_{-v})\}
\le
\left(
\sqrt{pq}+\sqrt{(1-p)(1-q)}
\right)^{M(S)},
\qquad
p=\frac{a\log n}{n},\quad q=\frac{b\log n}{n}.
\]
Since
\[
\sqrt{pq}+\sqrt{(1-p)(1-q)}
=1-\Dplus{a}{b}\frac{\log n}{n}+O\left(\frac{\log^2 n}{n^2}\right),
\]
we obtain
\begin{equation}
\label{eq:balanced_mle_alt_bound}
\Pr\{\text{$\tau$ has likelihood at least that of the truth}\}
\le
\exp\left(-(\Dplus{a}{b}-o(1))\frac{m(N-m)}{n}\log n\right).
\end{equation}

It remains to union bound over balanced alternatives.  The number of alternatives with Hamming distance $m$ is at most $\binom{N}{m}\le\exp(m\log(eN/m))$.  For $n^{1/4}\le m\le n/\log n$, we have $(N-m)/n=1-o(1)$ and $\log(eN/m)\le(3/4+o(1))\log n$, so the total contribution from this range is at most
\[
\sum_{m=n^{1/4}}^{n/\log n}
\exp\left(-\bigl(\Dplus{a}{b}-\tfrac34-o(1)\bigr)m\log n\right)=o(n^{-2}).
\]
For $n/\log n<m\le N/2$, we have $(N-m)/n\ge1/2-o(1)$ and $\log(eN/m)\le\log(e\log n)=O(\log\log n)$.  Therefore the likelihood exponent, which is of order $m\log n$, dominates the entropy term $m\log(eN/m)$, and this range contributes
\[
\sum_{m>n/\log n}
\exp\left(-\Omega(m\log n)\right)=o(n^{-2}).
\]
Thus, for this fixed $v$, with probability $1-o(n^{-2})$, no balanced labeling at distance at least $n^{1/4}$ has likelihood at least that of the truth.  A union bound over the $n$ choices of $v$ gives, with probability $1-o(1)$, that every leave-one-out maximizer is within $n^{1/4}$ of {either $\sigma^{*(-v)}$ or $-\sigma^{*(-v)}$}.

It remains only to check that the deterministic anchor convention chooses the same sign for all $v$.  Since $|D_n|=n^{3/4}+O(1)$ and the true labels are sampled from the exactly balanced model, hypergeometric anti-concentration gives
\[
\left|\sum_{u\in D_n}\sigma^*_u\right|>n^{1/3}
\]
with probability $1-o(1)$.  Removing one vertex changes this anchor sum by at most one, and the leave-one-out error bound above contributes at most $n^{1/4+o(1)}=o(n^{1/3})$ errors on $D_n$.  Hence the anchor orientation selects the same global sign for every leave-one-out maximizer.  The claimed simultaneous $n^{1/4+o(1)}$ bound follows.
\end{proof}

\concretefirststagescreening*

\begin{proof}
Let $G^{(0)}$ denote the first-stage observed graph, and let $A^{(0)}$ be its
adjacency matrix. 
For notational convenience in the proof, write
\[
    a:=16\theta_0,\qquad b:=\theta_0.
\]
We first identify the distribution of the first-stage graph.  The underlying graph is
\[
G\sim {\rm SBM}\left(n,16\frac{\log n}{n},\frac{\log n}{n}\right).
\]
Every vertex is queried exactly once.  Thus an existing edge \(\{u,v\}\) is missed by the first
stage if it is missed when querying \(u\) and also missed when querying \(v\).  These two oracle
coins are independent, so the miss probability is
\[
(1-w_0)^2=1-\theta_0.
\]
Therefore an existing edge is revealed with probability \(\theta_0\).  Since the oracle never
creates false edges, the first-stage observed graph is exactly an SBM with edge probabilities
\[
a\frac{\log n}{n}
=
16\theta_0\frac{\log n}{n},
\qquad
b\frac{\log n}{n}
=
\theta_0\frac{\log n}{n}.
\]
Equivalently,
\[
A^{(0)}
\sim
{\rm SBM}\left(n,a\frac{\log n}{n},b\frac{\log n}{n}\right).
\]

The useful point is that this graph is not strong enough for exact recovery, but it is strong
enough for rough recovery.  Indeed,
\[
D_+(a,b)
=
\frac12(\sqrt a-\sqrt b)^2
=
\frac12\theta_0(\sqrt{16}-\sqrt1)^2
=
\frac92(1-2^{-1/3})
>0.92>\frac34 .
\]
Hence \Cref{lem:rough_initialization} applies to the first-stage graph.  It gives balanced
leave-one-out rough labelings \(\widehat\sigma^{(-v)}\), one for each \(v\in[n]\), such that with
probability \(1-o(1)\), there is a common global sign \(s\in\{\pm1\}\) satisfying
\[
\max_{v\in[n]}
\left|
\left\{
u\neq v:
\widehat\sigma^{(-v)}_u\neq s\sigma^*_u
\right\}
\right|
\le n^{1/4+o(1)}
=o(n/\log n).
\]
Here \(\widehat\sigma^{(-v)}\) is computed after deleting all first-stage observations incident to
\(v\).  This leave-one-out construction is important: it makes the reference labels
\(\widehat\sigma^{(-v)}\) independent of the first-stage observations from \(v\), conditional on
the true labels.

For each \(v\), split the other vertices into two deterministic pieces.  List
\([n]\setminus\{v\}\) in increasing order.  Let \(B_v\) be the first
\(\lfloor(n-1)/2\rfloor\) vertices in this list, and let
\[
C_v:=([n]\setminus\{v\})\setminus B_v .
\]
The set \(C_v\) will be used only to decide whether \(v\) is ambiguous.  The set \(B_v\) will be
reserved for the second-stage likelihood test.  This split prevents us from using the same
incident observations both to screen \(v\) and later to classify \(v\).

Since the true labeling is sampled uniformly from the exactly balanced model, a deterministic
set of size \(n/2+O(1)\) contains \(n/4+o(n)\) vertices from each community, with high
probability.  By hypergeometric concentration and a union bound over all \(v\), we have
simultaneously for every \(v\),
\[
|B_v\cap\{u:\sigma^*_u=+1\}|=\frac n4+o(n),
\qquad
|B_v\cap\{u:\sigma^*_u=-1\}|=\frac n4+o(n),
\]
and also
\[
|C_v\cap\{u:\sigma^*_u=+1\}|=\frac n4+o(n),
\qquad
|C_v\cap\{u:\sigma^*_u=-1\}|=\frac n4+o(n).
\]

Now define the leave-one-out screening score
\[
\Lambda_v^{\rm loo}
:=
\sum_{u\in C_v}\widehat\sigma^{(-v)}_u A^{(0)}_{uv}.
\]
This score has a simple interpretation.  It compares how many revealed first-stage edges from
\(v\) go to vertices currently labeled \(+1\) and how many go to vertices currently labeled
\(-1\).  If \(v\) is truly in the \(+1\) community, then it should have more revealed edges to
true \(+1\) vertices than to true \(-1\) vertices, so the score should usually be positive.  If
\(v\) is truly in the \(-1\) community, the score should usually be negative.  Vertices with
scores close to zero are the ambiguous vertices that we will send to the second stage.

Fix the common global sign \(s\) from the rough-initialization guarantee.  For readability,
assume first that \(s=+1\) and \(\sigma^*_v=+1\).  If the reference labels
\(\widehat\sigma^{(-v)}\) were exact on \(C_v\), then \(\Lambda_v^{\rm loo}\) would have the
same large-deviation behavior as
\[
X_v-Y_v,
\]
where
\[
X_v\sim {\rm Bin}\left(\frac n4+o(n),a\frac{\log n}{n}\right),
\qquad
Y_v\sim {\rm Bin}\left(\frac n4+o(n),b\frac{\log n}{n}\right).
\]
Here \(X_v\) counts revealed edges from \(v\) to same-community vertices in \(C_v\), while
\(Y_v\) counts revealed edges from \(v\) to opposite-community vertices in \(C_v\).

The reference labels are not perfectly exact, but the number of reference-label errors in
\(C_v\) is \(o(n/\log n)\), uniformly over \(v\).  Conditional on the leave-one-out graph, the
number of first-stage revealed edges from \(v\) to these mislabeled reference vertices has mean
\[
O\left(\frac{\log n}{n}\right)o(n/\log n)=o(1).
\]

In the next claim, we show that the contribution of mislabeled reference vertices in $C_v$ to the leave-one-out screening score $\Lambda_v^{\rm loo}$ is extremely small. More precisely, for every $\varepsilon>0$, the probability that these mislabeled vertices perturb $\Lambda_v^{\rm loo}$ by more than $\varepsilon \log n$ is $o(n^{-2})$, uniformly over $v$.
\begin{innerclaim}[Bounding the tail event]
By Chernoff upper-tail bound, for every fixed \(\varepsilon>0\),
\[
\Pr\left( \text{the mislabeled reference vertices change } \Lambda_v^{\rm loo} \text{by more than }\varepsilon\log n \right) =o(n^{-2}),
\]
uniformly in $v$.
\end{innerclaim}
\begin{proof}
Let $D_v \subseteq C_v$ be the set of mislabeled reference vertices in $C_v$, and let $M_v$ be the number of first-stage revealed edges from $v$ to vertices in $D_v$. By the earlier part of the proof, uniformly over $v$, 
\[
|D_v| = o(n/ \log n).
\]
Since each edge $\{u,v\}$ with $u \in D_v$ is revealed in the first stage graph with probability at most $c \log n /n$ for some constant $c>0$, it follows that $M_v$ is stochastically dominated by
\[
\text{Bin} \left(|D_v|,c \frac{\log n}{n} \right),
\]
Hence, again uniformly over $v$,
\[ 
\lambda_v := \EE[M_v] = O\left(|D_v| \frac{\log n}{n} \right) = o(1).
\]
Now let
\[
|\Delta_v|:=|\Lambda^{\rm loo}_v - (X_v -Y_v)|
\]
so that $\Delta_v$ is the perturbation of the screening score caused by the mislabeled reference vertices in $C_v$. Each revealed edge to a mislabeled reference vertex changes the score by at most a constant; in fact, under the present definition of $\Lambda^{\rm loo}_v$, each such edge changes the score by at most $2$. Therefore, for some absolute constant $C>0$,
\[
|\Delta_v| \leq C M_v.
\]
Consequently, for every fixed $\varepsilon>0$,
\[
\Pr(|\Delta_v| > \varepsilon \log n) \leq \Pr\left(M_v > \frac{\varepsilon}{C} \log n\right).
\]
Set $t_v = \frac{\epsilon}{C} \log n$ and applying Chernoff upper-tail bound for the binomial random variable gives
\[
\Pr(M_v \geq t_v) \leq \left( \frac{e\lambda_v}{t_v}\right)^{t_v}
\]
Since $\lambda_v = o(1)$ uniformly over $v$ and $t_v = \Theta(\log n)$, we obtain
\[
\left(\frac{e\lambda_v}{t_v}\right)^{t_v} = \exp\left(- \Theta(\log n \log \log n)\right) = n^{-\omega(1)}.
\]
In particular, this is $o(n^{-2})$ uniformly over $v$.
\end{proof}

A union bound over \(v\) shows that, simultaneously for all vertices, the
effect of mislabeled reference vertices is only \(o(\log n)\) with probability $1-o(1)$.  Thus the large-deviation exponents of the actual score \(\Lambda_v^{\rm loo}\) are the same as those of \(X_v-Y_v\), up
to \(o(1)\) in the exponent.

Let \(I_{a,b}(d)\) be the Cram{\'e}r rate function for \(X_v-Y_v\) at scale \(\log n\):
\[
I_{a,b}(d)
:=
\sup_{t\in\mathbb R}
\left\{
td-\frac a4(e^t-1)-\frac b4(e^{-t}-1)
\right\}.
\]
This rate function describes the probability that \(X_v-Y_v\) has an atypically small value:
\[
\Pr(X_v-Y_v\le d\log n)
=
n^{-I_{a,b}(d)+o(1)}
\]
in the lower-tail regime relevant here.  In particular,
\[
I_{a,b}(0) = \frac14(\sqrt a-\sqrt b)^2 = \frac14\theta_0(\sqrt{16}-\sqrt1)^2
= \frac94(1-2^{-1/3}) >0.46.
\]
Choose
\[
\delta:=0.32.
\]
A direct evaluation of the displayed rate function gives
\[
I_{a,b}(\delta)>0.13, \qquad I_{a,b}(-\delta)>1.02.
\]

Define the observable screened set
\[
\mathcal A := \left\{ v:|\Lambda_v^{\rm loo}|\le \delta\log n \right\}.
\]
For vertices outside the screen, define the preliminary label by
\[
\widehat\sigma^{(0)}_v := {\rm sign}(\Lambda_v^{\rm loo}), \qquad v\notin\mathcal A.
\]
Thus \(\mathcal A\) contains exactly those vertices whose first-stage score has too small a
margin to trust.

We now bound the size of \(\mathcal A\).  Suppose \(\sigma^*_v=+1\).  Since the typical value of
\(\Lambda_v^{\rm loo}\) is positive, the event \(v\in\mathcal A\) is a lower-tail event:
\[
\Pr(v\in\mathcal A) = \Pr(|\Lambda_v^{\rm loo}|\le \delta\log n)
\le n^{-I_{a,b}(\delta)+o(1)} \le n^{-0.13+o(1)}.
\]
The same estimate holds when \(\sigma^*_v=-1\), by symmetry.  Therefore
\[
\mathbb E|\mathcal A| \le n^{0.87+o(1)}.
\]
By Markov's inequality,
\[
|\mathcal A|\le n^{0.88}
\]
with high probability.

Next we show that every vertex outside the screen is labeled correctly.  If \(\sigma^*_v=+1\),
a wrong sign outside the screen requires
\[
\Lambda_v^{\rm loo}<-\delta\log n.
\]
The corresponding large-deviation bound gives
\[
\Pr(\Lambda_v^{\rm loo}<-\delta\log n) \le n^{-I_{a,b}(-\delta)+o(1)}
\le n^{-1.02+o(1)}.
\]
Equivalently, for both possible true labels,
\[
\Pr(\sigma^*_v\Lambda_v^{\rm loo}<-\delta\log n) \le n^{-1.02+o(1)}.
\]
A union bound over all vertices gives
\[
\Pr\left( \exists v\notin\mathcal A: \widehat\sigma^{(0)}_v\neq s\sigma^*_v \right)
= o(1).
\]
Thus, with high probability, every vertex outside \(\mathcal A\) is correctly labeled, up to the same global sign \(s\).

It remains to check the reference condition for vertices inside the screen.  For each \(v\in\mathcal A\), we use \(B_v\) as the second-stage reference set and use \(\widehat\sigma^{(-v)}\) restricted to \(B_v\) as the reference labeling. The rough leave-one-out guarantee gives
\[
\left| \left\{ u\in B_v: \widehat\sigma^{(-v)}_u\neq s\sigma^*_u \right\} \right|
= o(n/\log n)
\]
simultaneously for every \(v\).  The hypergeometric concentration estimates above give
\[
|B_v\cap\{u:\sigma^*_u=+1\}|=\frac n4+o(n), \qquad |B_v\cap\{u:\sigma^*_u=-1\}|=\frac n4+o(n).
\]
Therefore \eqref{eq:reference_condition} holds with reference parameter
\[
\zeta=\frac12.
\]

Finally, the measurability and independence requirements in the definition of screened initialization hold by construction.  The event \(\{v\in\mathcal A\}\) uses first-stage observations incident to \(v\) only through vertices in \(C_v\).  The second-stage likelihood test will use fresh oracle queries from \(v\) to vertices in \(B_v\).  Since \(B_v\cap C_v=\varnothing\), screening does not condition on the first-stage observations from \(v\) to the reference set used later for refinement.  Also, the complement of \(B_v\cup\{v\}\) has linear size, so the exact-balance constraint does not determine \(\sigma^*_v\) from the labels on \(B_v\) alone.

Combining the preceding estimates, with probability \(1-o(1)\),
\[
|\mathcal A|\le n^{0.88}, \qquad \widehat\sigma^{(0)}_u=s\sigma^*_u
\quad\text{for all }u\notin\mathcal A,
\]
and for every \(v\in\mathcal A\), the leave-one-out reference set \(B_v\) satisfies \eqref{eq:reference_condition} with \(\zeta=1/2\).  Hence the first-stage graph admits a \(0.88\)-screening initialization with reference parameter \(1/2\).
\end{proof}

\sublinearoracleimpossible*
\begin{proof}
We prove the contrapositive. Fix an arbitrary adaptive oracle-only strategy
with query budget \(k=o(n)\). We show that exact recovery fails with probability
tending to one.

The main intuition is the following. A query at a vertex can only reveal neighbors of that queried vertex. If \(k=o(n)\), then only \(o(n)\) vertices can be queried at all. Moreover, since the graph has logarithmic maximum degree with high probability, the total number of vertex appearances in all oracle responses is only \(o(n\log n)\). Hence most vertices are both unqueried and appear only rarely in the entire transcript. Swapping the labels of two such vertices, one from each community, changes the likelihood of the transcript by only a small polynomial factor. Since there are linearly many such balanced swaps, the true labeling cannot dominate the posterior.

Let \(T\) denote the full oracle transcript, including the algorithm's internal
randomness. Define
\[
    Q(T):=\{v:\text{vertex }v\text{ is queried at least once}\}.
\]
Since the strategy makes at most \(k\) queries,
\[
    |Q(T)|\le k=o(n).
\]

We next control how often vertices appear in oracle responses. For each
\(v\in[n]\), let \(A_v(T)\) be the number of times \(v\) is returned as a
neighbor in the oracle transcript, counted with multiplicity over repeated
queries. Since the oracle never returns non-neighbors, every response is a
subset of the true neighborhood of the queried vertex.

With high probability, the underlying SBM has maximum degree at most
\(C_{\deg}\log n\), for a constant \(C_{\deg}=C_{\deg}(\alpha,\beta)\). On this
event, every oracle response has size at most \(C_{\deg}\log n\), and hence
\[
    \sum_{v=1}^n A_v(T) \le C_{\deg} k\log n = o(n\log n).
\]
Fix a small constant \(\varepsilon>0\), to be chosen later, and define the set
of vertices that appear too often in the transcript:
\[
    H(T):=\{v:A_v(T)>\varepsilon\log n\}.
\]
Then
\[
    |H(T)| \le \frac{1}{\varepsilon\log n}\sum_{v=1}^n A_v(T) \le \frac{C_{\deg}}{\varepsilon}k = o(n).
\]

Now define the good set
\[
    G(T):=[n]\setminus\bigl(Q(T)\cup H(T)\bigr).
\]
Thus a vertex in \(G(T)\) is never queried and is returned as a neighbor at most
\(\varepsilon\log n\) times in the entire transcript. Since both \(Q(T)\) and
\(H(T)\) have size \(o(n)\), we have
\[
    |G(T)|=n-o(n)
\]
with high probability. Because the true labeling is exactly balanced, removing
only \(o(n)\) vertices leaves linearly many vertices in each community:
\[
    |G(T)\cap V^+|=\frac n2-o(n), \qquad |G(T)\cap V^-|=\frac n2-o(n).
\]

We now fix two good vertices of opposite labels,
\[
    i\in G(T)\cap V^+, \qquad j\in G(T)\cap V^-.
\]
Let \(\sigma^{(i,j)}\) be the balanced labeling obtained by swapping the labels
of \(i\) and \(j\):
\[
    \sigma^{(i,j)}_i=-1, \qquad \sigma^{(i,j)}_j=+1, \qquad \sigma^{(i,j)}_u
    =\sigma^\star_u \quad\text{for }u\notin\{i,j\}.
\]
This is the natural local obstruction in the balanced SBM: a one-vertex flip is
not an admissible balanced labeling, but swapping one \(+\) vertex with one
\(-\) vertex is.

Define the likelihood ratio
\[
    L_{ij}(T) := \frac{\mathbb P_{\sigma^{(i,j)}}(T)} {\mathbb P_{\sigma^\star}(T)}.
\]

We now lower-bound \(L_{ij}(T)\). The important point is that we do this
directly on the original adaptive transcript. We do not compare the original run
to a counterfactual run. The algorithm's future query choices may depend on
past oracle responses, but once a transcript \(T\) is fixed, the query choices
are deterministic functions of the past transcript and the internal randomness.
Thus the query choices themselves contribute no likelihood factor. Only the
probabilities of the observed oracle responses matter.

Under the swap \(\sigma^\star\mapsto\sigma^{(i,j)}\), all pairs not incident to
\(i\) or \(j\) have the same within/between-community status as before. The pair
\(\{i,j\}\) is also between-community under both labelings, so it contributes no
likelihood difference. Therefore the only possible likelihood differences come
from pairs \(\{u,i\}\) and \(\{u,j\}\) with \(u\notin\{i,j\}\).

Since \(i,j\in G(T)\), neither \(i\) nor \(j\) is ever queried. Thus the only
way the transcript can see \(i\) or \(j\) is if some queried vertex \(u\) returns
\(i\) or \(j\) as a neighbor.

Consider one such pair \(\{u,v\}\), where \(v\in\{i,j\}\) and \(u\) is a
queried vertex. Suppose \(u\) is queried \(r\) times in the realized transcript.
If \(v\) is returned at least once in those \(r\) queries, then the likelihood
ratio contribution of the pair \(\{u,v\}\) is bounded below by a positive
constant depending only on \(\alpha,\beta\). Indeed, the two possible edge
probabilities are \(p\) and \(q\), so the return event changes likelihood by at
worst a factor comparable to \(p/q\) or \(q/p\).

If \(v\) is never returned in those \(r\) queries, then the probability of this
non-return event is
\[
    1-\rho p \qquad\text{or}\qquad 1-\rho q, \qquad \rho:=1-(1-w)^r\in[0,1],
\]
depending on whether the pair is within-community or between-community. Since
\(p,q=O(\log n/n)\), uniformly over \(\rho\in[0,1]\),
\[
    \left| \log\frac{1-\rho p}{1-\rho q} \right| = O\!\left(\frac{\log n}{n}\right).
\]
Thus a non-return contributes only \(O(\log n/n)\) log-likelihood loss.

Combining these two cases, there exists a constant
\(C_0=C_0(\alpha,\beta,w)\) such that, for the realized transcript,
\[
    \log L_{ij}(T) \ge -C_0\bigl(A_i(T)+A_j(T)\bigr) - C_0 |Q(T)|\frac{\log n}{n}.
\]
The first term accounts for actual returns of \(i\) or \(j\), each of which can
cost only a constant factor in likelihood. The second term accounts for all
queried vertices that did not return \(i\) or \(j\), each of which costs only
\(O(\log n/n)\).

Because \(i,j\in G(T)\),
\[
    A_i(T)\le \varepsilon\log n, \qquad A_j(T)\le \varepsilon\log n.
\]
Also \(|Q(T)|\le k=o(n)\), so
\[
    |Q(T)|\frac{\log n}{n}=o(\log n).
\]
Therefore
\[
    \log L_{ij}(T) \ge -2C_0\varepsilon\log n-o(\log n),
\]
or equivalently
\[
    L_{ij}(T) \ge n^{-2C_0\varepsilon-o(1)}.
\]

Choose \(\varepsilon>0\) small enough that
\[
    2C_0\varepsilon<\frac12.
\]
Then every balanced swap between two good vertices of opposite communities has
likelihood at least
\[
    L_{ij}(T)\ge n^{-1/2-o(1)}.
\]

On the high-probability event above, choose an arbitrary matching
\[
    \mathcal M(T) \subset \bigl(G(T)\cap V^+\bigr)\times\bigl(G(T)\cap V^-\bigr)
\]
of size
\[
    |\mathcal M(T)|=\frac n2-o(n).
\]
For every \((i,j)\in\mathcal M(T)\), the labeling \(\sigma^{(i,j)}\) is balanced
and distinct from \(\sigma^\star\), and satisfies
\[
    \frac{\mathbb P_{\sigma^{(i,j)}}(T)} {\mathbb P_{\sigma^\star}(T)} 
    = L_{ij}(T) \ge n^{-1/2-o(1)}.
\]

Since the prior over balanced labelings is uniform, the posterior mass assigned
to the true labeling is
\[
\begin{aligned}
    \Pi(\sigma^\star\mid T)
    &=
    \frac{\mathbb P_{\sigma^\star}(T)}
         {\sum_{\tau\ \mathrm{balanced}}\mathbb P_\tau(T)}  \\
    &\le
    \frac{\mathbb P_{\sigma^\star}(T)}
         {\mathbb P_{\sigma^\star}(T)
          +
          \sum_{(i,j)\in\mathcal M(T)}
          \mathbb P_{\sigma^{(i,j)}}(T)}  \\
    &=
    \frac{1}
         {1+\sum_{(i,j)\in\mathcal M(T)}L_{ij}(T)}  \\
    &\le
    \frac{1}
         {1+\left(\frac n2-o(n)\right)n^{-1/2-o(1)}}  \\
    &=o(1).
\end{aligned}
\]

Thus, with probability \(1-o(1)\), the true balanced labeling has vanishing
posterior mass. To translate this into an impossibility statement for every
estimator, recall the standard Bayes-risk identity. Conditional on the transcript
\(T\), the best possible estimator is the MAP estimator, whose conditional
success probability is
\[
    \max_{\tau\ \mathrm{balanced}}\Pi(\tau\mid T).
\]
Moreover,
\[
    \mathbb E\bigl[\Pi(\sigma^\star\mid T)\mid T\bigr]
    = \sum_{\tau\ \mathrm{balanced}}\Pi(\tau\mid T)^2.
\]
Since \(\Pi(\sigma^\star\mid T)=o(1)\) with high probability and is always at
most one, we have
\[
    \mathbb E\sum_{\tau\ \mathrm{balanced}}\Pi(\tau\mid T)^2=o(1).
\]
Therefore
\[
    \mathbb E\max_{\tau\ \mathrm{balanced}}\Pi(\tau\mid T)
    \le \mathbb E\left[ \left(\sum_{\tau\ \mathrm{balanced}}\Pi(\tau\mid T)^2\right)^{1/2} \right]
    =o(1).
\]
Hence even the MAP estimator succeeds with probability \(o(1)\), and no
estimator can achieve exact recovery with probability tending to one.

We conclude that every oracle-only exact-recovery strategy must use
\(k=\Omega(n)\) queries.
\end{proof}

\newpage
\section{Omitted proofs from \Cref{sec:recovery}} \label{app:sec:recovery}
\uniformsublinearfails*
\begin{proof}
The proof has one simple idea. When \(k=o(n)\), a balanced uniform
querying rule touches only \(o(n)\) vertices. Hence almost every vertex is
unqueried. For two unqueried vertices of opposite communities, the extra
oracle observations affect only \(o(n)\) possible comparison vertices, while
the \(n-o(n)\) remaining comparisons are exactly those from the subsampled
graph. Thus the local exact-recovery exponent for such vertices is the same
as for \(G_{\rm sub}\) alone.

Write
\[
    p=\alpha\frac{\log n}{n},
    \qquad
    q=\beta\frac{\log n}{n},
    \qquad
    I=(\sqrt{\alpha}-\sqrt{\beta})^2 .
\]
For all sufficiently large \(n\), since \(k=o(n)\), we have \(k<n\). A
balanced uniform non-adaptive strategy queries \(k\) vertices once and leaves
the remaining \(n-k\) vertices unqueried. Let
\[
    Q:=\{i\in[n]:R_i=1\},
    \qquad
    U:=[n]\setminus Q .
\]
Thus \(|Q|=k=o(n)\) and \(|U|=n-o(n)\). We condition throughout on the
queried set \(Q\). Since the querying rule is data-independent, \(Q\) is fixed
in advance, or randomized independently of both the graph and the labels.

For an existing edge \(\{i,j\}\), the edge is observed in the final data if it
is retained in \(G_{\rm sub}\), or if it is revealed by an oracle query to one
of its endpoints. Conditional on the query counts,
\[
    \rho_{ij}
    :=
    \Pr(\{i,j\}\text{ is observed}\mid \{i,j\}\in E)
    =
    1-(1-s)(1-w)^{R_i+R_j}.
\]
In particular, if \(i,j\in U\), then \(\rho_{ij}=s\), whereas if
\(i\in U\) and \(j\in Q\), then
\[
    \rho_{ij}=1-(1-s)(1-w)=s+(1-s)w .
\]
Therefore, for every unqueried vertex \(i\in U\),
\begin{align}
    \sum_{j\neq i}\rho_{ij}
    &=
    \sum_{j\in U\setminus\{i\}}\rho_{ij}
    +
    \sum_{j\in Q}\rho_{ij}                                      \\
    &=
    s(|U|-1)+\bigl(s+(1-s)w\bigr)|Q|                            \\
    &=
    s(n-k-1)+\bigl(s+(1-s)w\bigr)k                              \\
    &=
    sn+o(n).
    \label{eq:uniform-unqueried-exposure}
\end{align}
Thus an unqueried vertex has total exposure \(sn+o(n)\), and the extra
oracle exposure coming from queried vertices contributes only \(o(n)\) possible
comparison vertices.

We now prove impossibility using balanced two-vertex swaps. Let
\[
    U_+ := U\cap\{i:\sigma_i^\star=+1\},
    \qquad
    U_- := U\cap\{i:\sigma_i^\star=-1\}.
\]
Because \(Q\) is independent of the labels and \(|Q|=o(n)\), the exact
balance of the true labeling implies, with probability \(1-o(1)\),
\begin{equation}
    |U_+|=\frac n2-o(n),
    \qquad
    |U_-|=\frac n2-o(n).
    \label{eq:many-unqueried-both-signs}
\end{equation}
We work on this event.

Let \(Y\) denote the final observation, consisting of \(G_{\rm sub}\) and the
oracle-revealed edges. Conditional on \(Q\) and on a candidate labeling
\(\sigma\), the final observed edge indicators are independent over unordered
pairs, and
\[
    Y_{ab}\sim
    \begin{cases}
        \operatorname{Bernoulli}(\rho_{ab}p),
            & \sigma_a=\sigma_b,\\[1mm]
        \operatorname{Bernoulli}(\rho_{ab}q),
            & \sigma_a\neq\sigma_b .
    \end{cases}
\]

For \(i\in U_+\), define the signed one-vertex log-likelihood ratio in favor
of assigning \(i\) the wrong label:
\[
    X_i
    :=
    \sum_{\ell\neq i}
    \sigma_\ell^\star
    \left[
        Y_{i\ell}\log\frac{q}{p}
        +
        (1-Y_{i\ell})
        \log\frac{1-\rho_{i\ell}q}{1-\rho_{i\ell}p}
    \right].
\]
Indeed, if \(\sigma_\ell^\star=+1\), then under the true label
\(\sigma_i=+1\) the pair \(\{i,\ell\}\) is within-community, while under the
wrong label \(\sigma_i=-1\) it is between-community. The corresponding
single-pair log-likelihood ratio is
\[
   Y_{i\ell}\log\frac{\rho_{i\ell}q}{\rho_{i\ell}p}
   +(1-Y_{i\ell})\log\frac{1-\rho_{i\ell}q}{1-\rho_{i\ell}p}
   =
   Y_{i\ell}\log\frac{q}{p}
   +(1-Y_{i\ell})\log\frac{1-\rho_{i\ell}q}{1-\rho_{i\ell}p}.
\]
If instead \(\sigma_\ell^\star=-1\), the roles of \(p\) and \(q\) are
reversed, and the log-likelihood ratio changes sign. This is exactly captured
by the factor \(\sigma_\ell^\star\). Hence \(X_i\) is the log-likelihood ratio
between assigning \(i\) the wrong label \(-1\) and assigning it the true label
\(+1\), with all other labels fixed at their true values.

For \(j\in U_-\), define analogously
\[
    X_j
    :=
    -\sum_{\ell\neq j}
    \sigma_\ell^\star
    \left[
        Y_{j\ell}\log\frac{q}{p}
        +
        (1-Y_{j\ell})
        \log\frac{1-\rho_{j\ell}q}{1-\rho_{j\ell}p}
    \right],
\]
so that \(X_j\) is the log-likelihood ratio in favor of assigning \(j\) the
wrong label \(+1\).

We first estimate the probability that one such unqueried vertex has
atypically strong evidence for the wrong label. Fix \(i\in U_+\). Conditional
on the labels, \(X_i\) is a sum of independent terms. The terms corresponding
to queried vertices contribute negligibly to the logarithmic large-deviation
scale. More precisely, for any fixed \(\theta\) in a neighborhood of \(0\),
each \(\ell\in Q\) contributes
\[
   \log\mathbb E\exp\left\{
        \theta\,\sigma_\ell^\star
        \left[
            Y_{i\ell}\log\frac{q}{p}
            +(1-Y_{i\ell})
            \log\frac{1-\rho_{i\ell}q}{1-\rho_{i\ell}p}
        \right]
   \right\}
   =
   O\left(\frac{\log n}{n}\right),
\]
uniformly in \(\rho_{i\ell}\in[0,1]\). This is because
\(\Pr(Y_{i\ell}=1)=O(\log n/n)\), the edge-present log-ratio is \(O(1)\), and
the edge-absent log-ratio is \(O(\log n/n)\). Since \(|Q|=k=o(n)\), the total
contribution of queried vertices to \(\log\mathbb E e^{\theta X_i}\) is
\(o(\log n)\). Hence the queried vertices contribute \(o(1)\) to the
normalized cumulant generating function at scale \(\log n\). Equation
\eqref{eq:uniform-unqueried-exposure} also shows that their first-moment
contribution is \(o(n)\), so they do not shift the typical value of \(X_i\) at
the relevant scale.

The remaining \(n-o(n)\) comparison vertices are unqueried and therefore have
retention probability exactly \(s\). Thus the usual one-vertex Cramér
calculation for the logarithmic-degree SBM applies with effective parameters
\(s\alpha\) and \(s\beta\). Consequently,
\[
    \Pr_{\sigma^\star}(X_i\ge 0)
    =
    n^{-sI/2+o(1)}.
\]
More generally, for fixed \(t\) in a neighborhood of \(0\),
\[
    \Pr_{\sigma^\star}(X_i\ge t\log n)
    =
    n^{-J(t)+o(1)},
\]
where \(J\) is the corresponding Cramér rate function and
\[
    J(0)=\frac{sI}{2}.
\]
The matching lower bound is a standard strong large-deviation estimate for
triangular arrays. In particular, it follows from the Bahadur--Rao-type theorem
of Chaganty and Sethuraman~\cite{chaganty1993strong}, applied to the
log-moment generating functions of the independent Bernoulli summands. In the
logarithmic-degree SBM, the same logarithmic asymptotic is computed explicitly
in the converse proof of Abbe, Bandeira, and Hall~\cite{ABH16};
see also Mossel, Neeman, and Sly~\cite{mossel2015consistency} for the corresponding
exact-recovery converse. 

Since \(sI<2\), we have \(J(0)<1\). By continuity of \(J\), choose a
sufficiently small fixed \(\delta>0\) such that
\[
    \mu:=J(\delta)<1.
\]
Then, uniformly for \(i\in U_+\),
\begin{equation}
    \Pr_{\sigma^\star}(X_i\ge \delta\log n)
    =
    n^{-\mu+o(1)}.
    \label{eq:single-vertex-bad}
\end{equation}
The same estimate holds symmetrically for \(j\in U_-\).

Define
\[
    Z_+ := \bigl|\{i\in U_+: X_i\ge \delta\log n\}\bigr|,
    \qquad
    Z_- := \bigl|\{j\in U_-: X_j\ge \delta\log n\}\bigr|.
\]
By \eqref{eq:many-unqueried-both-signs} and
\eqref{eq:single-vertex-bad},
\[
    \mathbb E_{\sigma^\star} Z_+
    =
    n^{1-\mu+o(1)}
    \to\infty,
    \qquad
    \mathbb E_{\sigma^\star} Z_-
    =
    n^{1-\mu+o(1)}
    \to\infty .
\]

We next show concentration. For distinct \(i,i'\in U_+\), the score \(X_i\)
depends on the edge variables \(\{Y_{i\ell}:\ell\neq i\}\), while \(X_{i'}\)
depends on \(\{Y_{i'\ell}:\ell\neq i'\}\). These two collections share exactly
one unordered edge variable, namely \(Y_{ii'}\). All other edge observations
entering the two sums are disjoint and hence independent conditional on the
labels.

Write
\[
    X_i=S_i+T_{ii'},
    \qquad
    X_{i'}=S_{i'}+T_{i'i},
\]
where \(T_{ii'}\) and \(T_{i'i}\) are the contributions of the shared edge
\(\{i,i'\}\). Conditional on \(Y_{ii'}\), the variables \(S_i\) and
\(S_{i'}\) are independent. Moreover,
\[
    \Pr_{\sigma^\star}(Y_{ii'}=1)=O\left(\frac{\log n}{n}\right).
\]
If \(Y_{ii'}=0\), then \(T_{ii'},T_{i'i}=O(\log n/n)\), so the threshold
\(\delta\log n\) is shifted only by \(o(1)\). The sharp large-deviation
estimate above is stable under such an \(o(1)\) shift, giving the same tail
probability up to a factor \(1+o(1)\). If \(Y_{ii'}=1\), then
\(T_{ii'},T_{i'i}=O(1)\); this changes the logarithmic exponent only by
\(o(1)\), while the event \(Y_{ii'}=1\) itself has probability
\(O(\log n/n)\), making this contribution negligible compared to the product
of the two one-vertex tail probabilities. Therefore, uniformly over distinct
\(i,i'\in U_+\),
\[
    \Pr_{\sigma^\star}
    (X_i\ge \delta\log n,\ X_{i'}\ge \delta\log n)
    =
    (1+o(1))
    \Pr_{\sigma^\star}(X_i\ge \delta\log n)
    \Pr_{\sigma^\star}(X_{i'}\ge \delta\log n).
\]
It follows that
\[
    \operatorname{Var}(Z_+)=o((\mathbb E Z_+)^2).
\]
Hence, by Chebyshev's inequality, \(Z_+\ge 1\) with probability \(1-o(1)\).
The same argument gives \(Z_-\ge 1\) with probability \(1-o(1)\). By a union
bound, with probability \(1-o(1)\) both events occur.

On the event \(\{Z_+\ge 1,\ Z_-\ge 1\}\), choose
\(i\in U_+\) and \(j\in U_-\) such that
\[
    X_i\ge \delta\log n,
    \qquad
    X_j\ge \delta\log n.
\]
Let \(\sigma^{(i,j)}\) be the balanced labeling obtained from
\(\sigma^\star\) by swapping the labels of \(i\) and \(j\):
\[
    \sigma^{(i,j)}_i=-1,\qquad
    \sigma^{(i,j)}_j=+1,\qquad
    \sigma^{(i,j)}_\ell=\sigma^\star_\ell
    \quad\text{for all }\ell\notin\{i,j\}.
\]
Let
\[
    L_{ij}(Y)
    :=
    \frac{\mathbb P_{\sigma^{(i,j)}}(Y)}
         {\mathbb P_{\sigma^\star}(Y)}
\]
be the likelihood ratio between the swapped labeling and the true labeling.

Pairs \(\{u,v\}\) with \(u,v\notin\{i,j\}\) have identical distributions under
\(\sigma^\star\) and \(\sigma^{(i,j)}\), so they contribute \(0\) to
\(\log L_{ij}(Y)\). Pairs \(\{i,\ell\}\) with \(\ell\notin\{i,j\}\) contribute
the corresponding terms in \(X_i\), and pairs \(\{j,\ell\}\) with
\(\ell\notin\{i,j\}\) contribute the corresponding terms in \(X_j\). The pair
\(\{i,j\}\) is between-community under both \(\sigma^\star\) and
\(\sigma^{(i,j)}\), so it contributes \(0\) to the likelihood ratio of the
balanced swap.

The only bookkeeping issue is that the one-vertex scores \(X_i\) and \(X_j\)
include the terms corresponding to the pair \(\{i,j\}\), even though this pair
does not contribute to the balanced-swap likelihood ratio. Therefore
\[
    \log L_{ij}(Y)
    =
    X_i+X_j
    -
    \bigl[\text{the \(\ell=j\) term in \(X_i\)}\bigr]
    -
    \bigl[\text{the \(\ell=i\) term in \(X_j\)}\bigr].
\]
Each of the two subtracted terms is \(O(1)\) if \(Y_{ij}=1\), and
\(O(\log n/n)\) if \(Y_{ij}=0\). Hence, uniformly in \(Y\),
\[
    \log L_{ij}(Y)=X_i+X_j+O(1).
\]
Consequently,
\[
    \log L_{ij}(Y)
    \ge
    2\delta\log n-O(1)
    >
    0
\]
for all sufficiently large \(n\). Thus
\[
    \mathbb P_{\sigma^{(i,j)}}(Y)
    >
    \mathbb P_{\sigma^\star}(Y).
\]
Therefore, with probability \(1-o(1)\), there exists a balanced labeling
\(\sigma^{(i,j)}\neq\sigma^\star\) whose likelihood is strictly larger than
the likelihood of the true labeling. Since the prior over balanced labelings
is uniform, the maximum-a-posteriori estimator cannot return
\(\sigma^\star\) on this event. Hence exact recovery is impossible with
probability tending to one.

The argument was carried out after conditioning on the queried set \(Q\) and
on the high-probability event \eqref{eq:many-unqueried-both-signs}. Removing
this conditioning changes the failure probability by only \(o(1)\). The use
of data-independence of \(Q\) is essential in two places: first, it allows us
to condition on \(Q\) and retain a product measure over the edge observations
with deterministic exposures \(\rho_{ab}\); second, it gives
\eqref{eq:many-unqueried-both-signs}, since \(Q\) is independent of the true
labels. Thus the conclusion applies to any non-adaptive allocation that is
fixed in advance or randomized independently of \((G_{\rm sub},\sigma^\star)\).
It does not apply to allocations that depend on \(G_{\rm sub}\).

We conclude that exact recovery under balanced uniform non-adaptive querying
is impossible with high probability whenever
\[
    s(\sqrt{\alpha}-\sqrt{\beta})^2=sI<2 .
\]
This proves the lemma.
\end{proof}

\adaptivesublinearsucceed*
\begin{proof}
	All label recovery statements are understood up to the unavoidable global
	sign flip. We prove the lemma by giving one concrete choice of parameters and
	one concrete two-stage strategy.
	
	\paragraph{Step 1: Choose parameters.}
	Set
	\[
	\alpha=16,\qquad
	\beta=1,\qquad
	w=\frac12,
	\qquad
	(\sqrt{\alpha}-\sqrt{\beta})^2=9.
	\]
	Choose a sufficiently small fixed $\epsilon>0$ and set
	\[
	s=\frac{2(1-\epsilon)}{9}.
	\]
	Then the subsampled graph $G_{\rm sub}$ has effective exact-recovery signal
	\[
	s(\sqrt{\alpha}-\sqrt{\beta})^2
	=
	2(1-\epsilon)
	<2.
	\]
	Thus $G_{\rm sub}$ alone is below the sharp Abbe--Bandeira--Hall threshold
	for exact recovery in the logarithmic-degree SBM. This is useful context:
	the subsampled graph alone cannot exactly recover, so the oracle queries must
	be targeted to repair the remaining ambiguous vertices.
	
	Let
	\[
	c:=\frac{(\sqrt{\alpha}-\sqrt{\beta})^2}{2}
	=
	\frac92.
	\]
	Then
	\[
	cs=1-\epsilon.
	\]
	
	\paragraph{Step 2: A leave-one-out near-exact initial estimator.}
	We use the following standard near-exact recovery fact for the
	logarithmic-degree SBM. There exist estimators
	$\widehat\sigma^{(0,-v)}$, one for each $v\in[n]$, such that
	$\widehat\sigma^{(0,-v)}$ is computed from $G_{\rm sub}$ after deleting all
	edges incident to $v$, and such that, after orienting all estimators by a
	common deterministic anchor convention,
	\begin{equation}
		\label{eq:loo_near_exact_appendix}
		\max_{v\in[n]}
		\mathrm{Ham}\bigl(\widehat\sigma^{(0,-v)},\sigma^\star\bigr)
		\le
		n^{1-cs+o(1)}
		=
		n^{\epsilon+o(1)}
	\end{equation}
	with high probability.
	
	Intuitively, $G_{\rm sub}$ is below the exact-recovery threshold but still
	contains enough information to recover all but a subpolynomially small
	fraction of the labels. The leave-one-out construction is used only to
	ensure independence: $\widehat\sigma^{(0,-v)}$ is computed without looking at
	any subsampled edge incident to $v$, so it is independent of those incident
	edges conditional on the true labels.
	
	\paragraph{Step 3: Split the neighbors of each vertex into a screening set
		and a refinement set.}
	We now separate the information used to decide whether $v$ is ambiguous from
	the information used later to refine $v$ by oracle queries. This separation is
	the key technical point.
	
	Set
	\[
	b_0:=\frac{7}{20},
	\qquad
	c_0:=\frac{13}{20}.
	\]
	For every $v$, choose deterministic disjoint sets
	$B_v,C_v\subset[n]\setminus\{v\}$ satisfying
	\[
	|B_v|=b_0n+o(n),
	\qquad
	|C_v|=c_0n+o(n).
	\]
	For example, one can choose these sets by a fixed cyclic rule independent of
	the graph and of the labels. The set $C_v$ will be used only to screen $v$,
	and the set $B_v$ will be used only for the second-stage oracle refinement.
	
	Because the true labels are balanced and the sets $B_v,C_v$ are fixed
	independently of the labels, hypergeometric concentration and a union bound
	over $v$ imply that, with high probability, simultaneously for all $v$,
	\[
	|B_v\cap\{u:\sigma_u^\star=+1\}|
	=
	\frac{b_0n}{2}+o(n),
	\qquad
	|B_v\cap\{u:\sigma_u^\star=-1\}|
	=
	\frac{b_0n}{2}+o(n),
	\]
	and similarly
	\[
	|C_v\cap\{u:\sigma_u^\star=+1\}|
	=
	\frac{c_0n}{2}+o(n),
	\qquad
	|C_v\cap\{u:\sigma_u^\star=-1\}|
	=
	\frac{c_0n}{2}+o(n).
	\]
	
	\paragraph{Step 4: Define the screening score.}
	Write
	\[
	p=\alpha\frac{\log n}{n},
	\qquad
	q=\beta\frac{\log n}{n}.
	\]
	For each vertex $v$, define the leave-one-out screening score
	\[
	\Lambda_v^{\rm loo}
	:=
	\sum_{u\in C_v}
	\widehat\sigma^{(0,-v)}_u
	\left[
	A^{\rm sub}_{uv}\log\frac pq
	+
	(1-A^{\rm sub}_{uv})
	\log\frac{1-sp}{1-sq}
	\right].
	\]
	This is the log-likelihood evidence, computed from the subsampled graph, for
	whether $v$ has label $+1$ or $-1$, using the leave-one-out estimates of the
	labels of vertices in $C_v$ as references.
	
	We screen vertices whose evidence is not strong. Let
	\[
	\tau:=\frac{21}{20}
	\]
	and define
	\[
	\mathcal A
	:=
	\left\{
	v:
	|\Lambda_v^{\rm loo}|\le \tau\log n
	\right\}.
	\]
	The intended interpretation is simple: vertices outside $\mathcal A$ have a
	large positive or negative likelihood score and will be trusted; vertices in
	$\mathcal A$ are ambiguous and will be queried.
	
	\paragraph{Step 5: Analyze the ideal genie-aided screening score.}
	First replace the estimated reference labels
	$\widehat\sigma^{(0,-v)}_u$ by the true labels $\sigma_u^\star$. Define
	\[
	\Lambda_v^\star
	:=
	\sum_{u\in C_v}
	\sigma^\star_u
	\left[
	A^{\rm sub}_{uv}\log\frac pq
	+
	(1-A^{\rm sub}_{uv})
	\log\frac{1-sp}{1-sq}
	\right].
	\]
	We compute the large-deviation rate of $\Lambda_v^\star$ conditional on
	$\sigma_v^\star=+1$.
	
	For a reference vertex $u$ with $\sigma_u^\star=+1$, the subsampled edge
	indicator satisfies
	\[
	A^{\rm sub}_{uv}\sim \operatorname{Bernoulli}(sp).
	\]
	For such a vertex, the single-edge log-likelihood contribution is
	\[
	X_p
	=
	A^{\rm sub}_{uv}\log\frac pq
	+
	(1-A^{\rm sub}_{uv})
	\log\frac{1-sp}{1-sq}.
	\]
	For fixed $\theta$, using
	$p=\alpha\log n/n$ and $q=\beta\log n/n$, we have
	\[
	\log \mathbb E e^{\theta X_p}
	=
	\frac{s\log n}{n}
	\left[
	\alpha\left(\left(\frac{\alpha}{\beta}\right)^\theta-1-\theta\right)
	+\theta\beta
	\right]
	+
	o\left(\frac{\log n}{n}\right).
	\]
	For a reference vertex $u$ with $\sigma_u^\star=-1$, we have
	$A^{\rm sub}_{uv}\sim\operatorname{Bernoulli}(sq)$ and the score contributes
	the negative of the same log-likelihood expression. Hence
	\[
	\log \mathbb E e^{\theta X_q}
	=
	\frac{s\log n}{n}
	\left[
	\beta\left(\left(\frac{\beta}{\alpha}\right)^\theta-1-\theta\right)
	+\theta\alpha
	\right]
	+
	o\left(\frac{\log n}{n}\right).
	\]
	The linear terms in $\theta$ cancel when we sum over asymptotically equal
	numbers of $+1$ and $-1$ reference labels in $C_v$. Since
	$|C_v|=c_0n+o(n)$ and $C_v$ contains asymptotically half of each community,
	the normalized cumulant generating function of $\Lambda_v^\star$ is
	\[
	\psi_{c_0}(\theta)
	:=
	\lim_{n\to\infty}
	\frac{1}{\log n}
	\log
	\mathbb E\left[
	e^{\theta\Lambda_v^\star}
	\,\middle|\,
	\sigma_v^\star=+1
	\right]
	\]
	with
	\[
	\psi_{c_0}(\theta)
	=
	c_0\frac{s}{2}
	\left[
	\alpha\left(\left(\frac{\alpha}{\beta}\right)^\theta-1\right)
	+
	\beta\left(\left(\frac{\beta}{\alpha}\right)^\theta-1\right)
	\right].
	\]
	Let
	\[
	I_{c_0}(t)
	:=
	\sup_{\theta\in\mathbb R}
	\{\theta t-\psi_{c_0}(\theta)\}
	\]
	be the associated rate function. The Chernoff bound gives, for fixed $t$
	below the typical value of $\Lambda_v^\star/\log n$,
	\[
	\Pr\left(
	\Lambda_v^\star\le t\log n
	\,\middle|\,
	\sigma_v^\star=+1
	\right)
	\le
	n^{-I_{c_0}(t)+o(1)}.
	\]
	
	We now evaluate this rate function for our constants. At $\epsilon=0$,
	$s=2/9$, and therefore
	\[
	\psi_{c_0}(\theta)
	=
	\frac{13}{180}
	\left[
	16(16^\theta-1)+16^{-\theta}-1
	\right].
	\]
	The optimizer in
	\[
	I_{c_0}(t)
	=
	\sup_{\theta\in\mathbb R}
	\{\theta t-\psi_{c_0}(\theta)\}
	\]
	is the unique solution of $\psi_{c_0}'(\theta)=t$. Since
	\[
	\psi_{c_0}'(\theta)
	=
	\frac{13}{180}\log 16
	\left[
	16\cdot 16^\theta-16^{-\theta}
	\right],
	\]
	this is a one-dimensional equation. Direct evaluation at
	$t=\tau=21/20$ and $t=-\tau=-21/20$ gives
	\[
	I_{c_0}\left(\frac{21}{20}\right)
	\approx 0.24515,
	\qquad
	I_{c_0}\left(-\frac{21}{20}\right)
	\approx 1.29515.
	\]
	Hence, for $\epsilon>0$ sufficiently small, by continuity of the rate
	function in $s=2(1-\epsilon)/9$,
	\begin{equation}
		\label{eq:screening_rate_bounds_appendix}
		I_{c_0}(\tau)>0.24,
		\qquad
		I_{c_0}(-\tau)>1.25.
	\end{equation}
	
	\paragraph{Step 6: Replace the genie labels by the leave-one-out labels.}
	We now compare $\Lambda_v^{\rm loo}$ with $\Lambda_v^\star$. By
	\eqref{eq:loo_near_exact_appendix}, uniformly over $v$, the number of
	reference labels in $C_v$ on which
	$\widehat\sigma^{(0,-v)}$ disagrees with $\sigma^\star$ is at most
	$n^{\epsilon+o(1)}$.
	
	The effect of these wrong reference labels is negligible on the
	$\log n$ scale. Indeed, for any fixed $\theta$, changing the sign of a single
	reference label changes the corresponding single-vertex log-moment
	generating function by at most $O(\log n/n)$, because the edge itself appears
	with probability only $O(\log n/n)$. Therefore the total change in the
	log-moment generating function caused by all wrong reference labels is at
	most
	\[
	n^{\epsilon+o(1)}
	O\left(\frac{\log n}{n}\right)
	=
	o(1).
	\]
	In particular, the large-deviation exponents for $\Lambda_v^{\rm loo}$ are
	the same as those for the genie-aided score $\Lambda_v^\star$, up to $o(1)$.
	
	Combining this perturbation bound with
	\eqref{eq:screening_rate_bounds_appendix}, we get
	\[
	\Pr\left(
	|\Lambda_v^{\rm loo}|\le \tau\log n
	\right)
	\le
	n^{-0.24+o(1)}.
	\]
	Indeed, conditional on $\sigma_v^\star=+1$, the event
	$|\Lambda_v^{\rm loo}|\le\tau\log n$ is contained in the lower-tail event
	$\Lambda_v^{\rm loo}\le\tau\log n$, and the case
	$\sigma_v^\star=-1$ is symmetric.
	
	Similarly,
	\[
	\Pr\left(
	\Lambda_v^{\rm loo}\le -\tau\log n
	\,\middle|\,
	\sigma_v^\star=+1
	\right)
	\le
	n^{-1.25+o(1)},
	\]
	and, by symmetry,
	\[
	\Pr\left(
	\Lambda_v^{\rm loo}\ge \tau\log n
	\,\middle|\,
	\sigma_v^\star=-1
	\right)
	\le
	n^{-1.25+o(1)}.
	\]
	
	\paragraph{Step 7: The screen is small and all unscreened vertices are
		correct.}
	Since
	\[
	\mathcal A
	=
	\{v:|\Lambda_v^{\rm loo}|\le\tau\log n\},
	\]
	the previous bound gives
	\[
	\mathbb E|\mathcal A|
	\le
	n^{0.76+o(1)}.
	\]
	By Markov's inequality,
	\begin{equation}
		\label{eq:screen_size_appendix}
		|\mathcal A|\le n^{0.77}
	\end{equation}
	with high probability.
	
	Next consider a vertex outside the screen. If $\sigma_v^\star=+1$ and
	$v\notin\mathcal A$ but the score has the wrong sign, then necessarily
	\[
	\Lambda_v^{\rm loo}<-\tau\log n.
	\]
	This event has probability at most $n^{-1.25+o(1)}$. The same argument
	applies symmetrically when $\sigma_v^\star=-1$. A union bound over all
	vertices gives
	\[
	\Pr\left(
	\exists v\notin\mathcal A:
	\operatorname{sign}(\Lambda_v^{\rm loo})
	\neq
	\sigma_v^\star
	\right)
	\le
	n\cdot n^{-1.25+o(1)}
	=
	n^{-0.25+o(1)}
	=
	o(1).
	\]
	Thus, with high probability, every vertex outside $\mathcal A$ is already
	correctly labeled by $\operatorname{sign}(\Lambda_v^{\rm loo})$.
	
	\paragraph{Step 8: Refine the screened vertices using fresh oracle queries.}
	We now query every vertex in $\mathcal A$ once. The second-stage classifier
	for $v\in\mathcal A$ uses only oracle responses from $v$ to vertices in
	$B_v$, together with the reference labels
	$\widehat\sigma^{(0,-v)}_u$ for $u\in B_v$.
	
	The separation between $B_v$ and $C_v$ is crucial. The event
	$\{v\in\mathcal A\}$ depends on subsampled incident edges from $v$ only
	through pairs $\{u,v\}$ with $u\in C_v$. The second-stage oracle likelihood
	test uses only pairs $\{u,v\}$ with $u\in B_v$, and $B_v\cap C_v=\emptyset$.
	Moreover, $\widehat\sigma^{(0,-v)}$ was computed after deleting all
	subsampled edges incident to $v$. Therefore, conditional on the first-stage
	data and on the event $v\in\mathcal A$, the underlying edges from $v$ to
	$B_v$ are still independent SBM edges, and the oracle coins used in the
	second stage are fresh.
	
	For one oracle query, an edge from $v$ to $u$ is revealed with probability
	$w$ if it exists. Thus the effective second-stage edge probabilities are
	\[
	wp=w\alpha\frac{\log n}{n},
	\qquad
	wq=w\beta\frac{\log n}{n}.
	\]
	Using the same single-vertex Bhattacharyya calculation as in the standard
	SBM exact-recovery analysis, but now restricted to the reference set $B_v$,
	the error probability for classifying a fixed screened vertex is
	\[
	\Pr\left(
	\widehat\sigma_v^{(1)}\neq\sigma_v^\star
	\,\middle|\,
	v\in\mathcal A
	\right)
	\le
	n^{-b_0 w(\sqrt{\alpha}-\sqrt{\beta})^2/2+o(1)}.
	\]
	Here $B_v$ has asymptotic size $b_0n$ and contains asymptotically half of
	each community.
	
	The leave-one-out reference labels on $B_v$ are not perfectly correct, but
	by \eqref{eq:loo_near_exact_appendix} they contain at most
	$n^{\epsilon+o(1)}=o(n/\log n)$ errors, uniformly in $v$. As before, these
	wrong reference labels perturb the likelihood exponent by only $o(1)$ and
	do not change the displayed exponent.
	
	For our constants,
	\[
	\frac{b_0w}{2}(\sqrt{\alpha}-\sqrt{\beta})^2
	=
	\frac{7/20}{2}\cdot\frac12\cdot 9
	=
	\frac{63}{80}
	=
	0.7875.
	\]
	Therefore
	\[
	\Pr\left(
	\widehat\sigma_v^{(1)}\neq\sigma_v^\star
	\,\middle|\,
	v\in\mathcal A
	\right)
	\le
	n^{-0.7875+o(1)}.
	\]
	
	On the high-probability event \eqref{eq:screen_size_appendix},
	$|\mathcal A|\le n^{0.77}$. Hence a union bound over all screened vertices
	gives
	\[
	\Pr\left(
	\exists v\in\mathcal A:
	\widehat\sigma_v^{(1)}\neq\sigma_v^\star
	\right)
	\le
	n^{0.77}\cdot n^{-0.7875+o(1)}
	=
	n^{-0.0175+o(1)}
	=
	o(1).
	\]
	Thus all screened vertices are correctly recovered after one oracle query
	each.
	
	\paragraph{Step 9: Query budget and conclusion.}
	Choose
	\[
	k=\left\lceil n^{0.78}\right\rceil.
	\]
	By \eqref{eq:screen_size_appendix}, with high probability this budget is
	enough to query every vertex in $\mathcal A$ once. Since $0.78<1$, we have
	$k=o(n)$.
	
	With high probability, all vertices outside $\mathcal A$ are correctly
	labeled by the screening score, and all vertices inside $\mathcal A$ are
	correctly labeled by the second-stage oracle refinement. Therefore the
	two-stage targeted strategy achieves exact recovery with probability
	$1-o(1)$ using $k=o(n)$ oracle queries.
\end{proof}

\adaptivemain*
\begin{proof}
By \Cref{lemma:uniform_sublinear_fails}, when $k=o(n)$, uniform non-adaptive querying does not improve the exact-recovery threshold beyond that of the subsampled graph $G_{\rm sub}$ alone. Hence, for parameter choices satisfying
$s(\sqrt{\alpha}-\sqrt{\beta})^2<2$,
any uniform querying strategy fails with high probability. 

On the other hand, by \Cref{lemma:adaptive_sublinear_succeed}, there exist constants $(\alpha,\beta,s,w)$ and a sublinear budget $k=o(n)$ for which a two-stage strategy that uses $G_{\rm sub}$ to target its queries succeeds at exact recovery with high probability. Therefore, for the same observation model and the same asymptotic budget regime, uniform querying fails while a $G_{\rm sub}$-dependent non-uniform strategy succeeds. This establishes a strict adaptivity gap for \Cref{problem2}. 
\end{proof}

\end{document}